\documentclass[envcountsect,envcountsame]{llncs}
\usepackage{amsmath}
\usepackage{amssymb}
\usepackage{xspace}
\usepackage{color}
\usepackage{amsfonts}
\usepackage[T1]{fontenc}
\usepackage{url}
\usepackage{float}
\usepackage{framed}

\usepackage[a4paper,hmargin=1in,vmargin=1.1in]{geometry}
\usepackage{epsfig}
\usepackage{ifpdf}
\usepackage{graphicx}
\usepackage{arydshln}

\ifpdf

\else

\fi

\newcommand{\set}[1]{\{#1\}}
\newcommand{\Set}[2]{\{#1\,|\,#2\}}

\newcommand{\bin}{\set{0,1}} 

\newcommand{\N}{\Bbb{N}}

\newcommand{\R}{\Bbb{R}}

\newcommand{\eps}{\varepsilon}

\newcommand{\bhat}{\hat{b}}
\newcommand{\etal}{{\it et~al}}

\newcommand{\rndc}{\xi}
\newcommand{\rndo}{\eta}
\newcommand{\drndc}{\bar{\xi}}
\newcommand{\drndo}{\bar{\eta}}

\allowdisplaybreaks

\title{On the Composition of Two-Prover Commitments, and Applications to Multi-Round Relativistic Commitments\thanks{This paper is an extended version of
our EUROCRPYT 2016 paper. The eprint version is available at \texttt{https://eprint.iacr.org/2016/113}.}}

\author{Serge Fehr \and Max Fillinger\thanks{Supported by the {\em NWO Free Competition} grant 617.001.203.}}

\institute{Centrum Wiskunde \& Informatica (CWI), Amsterdam, The Netherlands  \\
{\tt \{serge.fehr,\,max.fillinger\}@cwi.nl}
}

\def\prot{{\sf prot}}

\def\com{{\sf com}}
\def\open{{\sf open}}

\def\dcom{\overline{\sf com}}
\def\dopen{\overline{\sf open}}
\def\dnewopen{\overline{\sf newopen}}

\def\ptoq{{\sf ptoq}}
\def\Extr{{\rm Extr}}
\def\CHSH{{\cal CHSH}}
\def\XCHSH{{\cal XCHSH}}

\begin{document}

{\def\addcontentsline#1#2#3{}\maketitle}


\begin{abstract}
We consider the related notions of {\em two-prover} and of {\em relativistic} commitment schemes. In recent work, Lunghi \etal. proposed a new relativistic commitment scheme with a {\em multi-round sustain phase} that keeps the binding property alive as long as the sustain phase is running. 
They prove security of their scheme against classical attacks; however, the proven bound on the error parameter is very weak: It blows up {\em double exponentially} in the number of rounds. 

In this work, we give a new analysis of the multi-round scheme of Lunghi \etal., and we show a {\em linear} growth of the error parameter instead (also considering classical attacks only). Our analysis is based on a new {\em composition theorem} for two-prover commitment schemes. The proof of our composition theorem is based on a better understanding of the binding property of two-prover commitments that we provide in the form of new definitions and relations among them. These new insights are certainly of independent interest and are likely to be useful in other contexts as well.

Finally, our work gives rise to several interesting open problems, for instance extending our results to the quantum setting, where the dishonest provers are allowed to perform measurements on an entangled quantum state in order to try to break the binding property. 
\end{abstract}

\tableofcontents

\newpage

\section{Introduction}

\paragraph{\sc Two-Prover Commitment Schemes. }
We consider the notion of {\em 2-prover commitment schemes}, as originally introduced by Ben-Or, Goldwasser, Kilian and Wigderson in their seminal paper~\cite{BGKW88}. In a 2-prover commitment scheme, the prover (i.e., the entity that is responsible for preparing and opening the commitment) consists of two agents, $P$ and $Q$, and it is assumed that these two agents cannot communicate with each other during the execution of the protocol. With this approach, the classical and quantum impossibility results for unconditionally secure commitment schemes \cite{Mayers97,LC97} can be circumvented. 

A simple 2-prover bit commitment scheme is the scheme proposed by Cr{\'e}peau \etal.~\cite{CSST11}, which works as follows. 
The verifier $V$ chooses a uniformly random $a \in \set{0,1}^n$ and sends it to $P$, who replies with $x := y + a \cdot b$, where $b$ is the bit to commit to, and $y \in \set{0,1}^n$ is a uniformly random string known (only) to $P$ and $Q$. Furthermore, ``$+$'' is bit-wise XOR, and ``$\cdot$'' is scalar multiplication (of the scalar $b$ with the vector~$a$). In order to open the commitment (to $b$), $Q$ sends $y$ to $V$, and $V$ checks if $x + y = a \cdot b$. It is clear that this scheme is hiding: The commitment $x = y + a \cdot b$ is uniformly random and independent of $a$ no matter what $b$ is. On the other hand, the binding property follows from the observation that in order to open the commitment to $b = 0$, $Q$ needs to announce $y = x$, and in order to open to $b = 1$, he needs to announce $y = x + a$. Thus, in order to open to {\em both}, he must know $x$ {\em and} $x + a$, and thus $a$, which is a contradiction to the no-communication assumption, because $a$ was sent to 
$P$ 
only. 

In the quantum setting, where the dishonest provers are allowed to share an entangled quantum state and can produce $x$ and $y$ by means of performing measurements on their respective parts of the state, the above reasoning for the binding property does not work anymore. Nevertheless, as shown in~\cite{CSST11}, the binding property still holds (though with a weaker parameter). 

\paragraph{\sc Relativistic Commitment Schemes. }
The idea of {\em relativistic commitment schemes}, as introduced by Kent~\cite{Kent99}, is to take a 2-prover commitment scheme as above and enforce the no-communication assumption by means of relativistic effects: Place $P$ and $Q$ spatially far apart, and execute the scheme fast enough, so that 
there is not enough time for them to communicate. The obvious downside of such a relativistic commitment scheme is that the binding property stays alive only for a very short time: The opening has to take place almost immediately after the committing, before the provers have the chance to exchange information. 
This limitation can be circumvented by considering {\em multi-round} schemes, where after the actual commit phase there is a {\em sustain phase}, during which the provers and the verifier keep exchanging messages, and as long as this sustain phase is running, the commitment stays binding (and hiding), until the commitment is finally opened. Such schemes were proposed in \cite{Kent99} and~\cite{Kent05}, but they are rather inefficient, and the security analyses are somewhat informal (e.g., with no formal security definitions) and of asymptotic nature.
Schemes that require quantum communication were also considered and studied
\cite{Kent12,KTHW13,LKB+13} but those were all without sustain
phase.

More recently, Lunghi \etal.~\cite{LKB+15} proposed a new and simple multi-round relativistic commitment scheme, and provided a rigorous security analysis. 
Their scheme works as follows. 
The actual commit protocol is the commit protocol from the Cr{\'e}peau \etal. scheme: $V$ sends a uniformly random string $a_0 \in \set{0,1}^n$ to $P$, who returns $x_0 := y_0 + a_0 \cdot b$. Then, to sustain the commitment, before $P$ has the chance to tell $a_0$ to $Q$, $V$ sends a new uniformly random string $a_1 \in \set{0,1}^n$ to $Q$ who replies with $x_1 := y_1 + a_1 \cdot y_0$, where $y_1 \in \set{0,1}^n$ is another random string shared between $P$ and $Q$, and the multiplication $a_1 \cdot y_0$ is in a suitable finite field. Then, to further sustain the commitment, $V$
sends a new uniformly random string $a_2 \in \set{0,1}^n$ to $P$ who replies with $x_2 := y_2 + a_2 \cdot y_1$, etc. Finally, after the last sustain round where $x_m := y_m + a_m \cdot y_{m-1}$ has been sent to $V$, in order to finally open the commitment, $y_m$ is sent to $V$ (by the other prover). See Figure~\ref{fig:multiround}. In order to verify the opening, $V$ computes $y_{m-1},y_{m-2},\ldots,y_0$ inductively in the obvious way, and checks if $x_0 + y_0 = a_0 \cdot b$. 

\begin{figure}
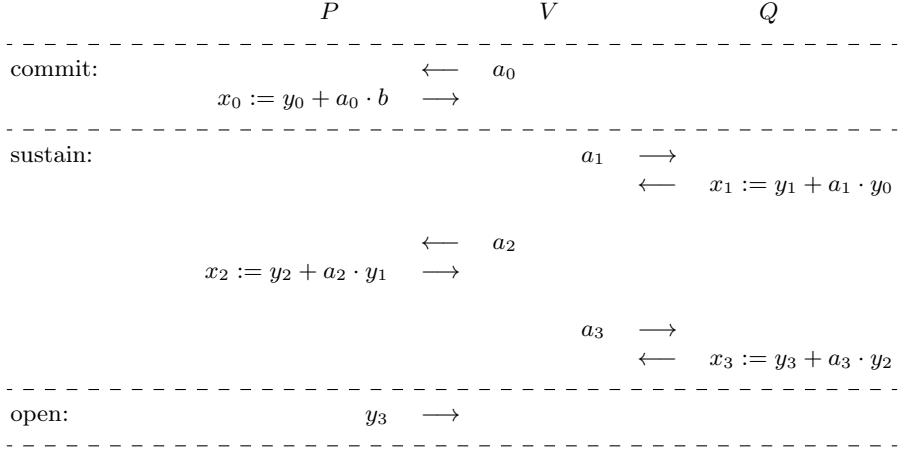

$$
\begin{array}{lrcccccl}
 & P \qquad & & & V & & & \qquad Q \\[1.6ex]
\hdashline \\[-1.8ex]
\text{commit:} \qquad\qquad & & \quad\longleftarrow\quad & a_0  \\
& x_0 := y_0 + a_0 \cdot b & \longrightarrow & \\[1.5ex]
\hdashline \\[-1.8ex]
\text{sustain:}& & & & \qquad & a_1 & \quad\longrightarrow\quad &  \\
& & & & \qquad &  & \quad\longleftarrow\quad & x_1 := y_1 + a_1 \cdot y_0 \\[2.8ex]
&  & \quad\longleftarrow\quad & a_2  \\
& x_2 := y_2 + a_2 \cdot y_1 & \longrightarrow & \\[2.8ex]
& & & & \qquad & a_3 & \quad\longrightarrow\quad &  \\
& & & & \qquad &  & \quad\longleftarrow\quad & x_3 := y_3 + a_3 \cdot y_2 \\[1.5ex]
\hdashline \\[-1.8ex]
\text{open:}& y_3 & \longrightarrow \\[1.5ex]
\hdashline 
\end{array}
$$
\caption{The Lunghi \etal. multi-round scheme (for $m = 3$). }\label{fig:multiround}
\end{figure}

What is crucial is that in round $i$ (say for odd $i$), when preparing $x_i$, the prover $Q$ must not know $a_{i-1}$, but he is allowed to know $a_1,\ldots,a_{i-2}$. Thus, the execution must be timed in such a way that between subsequent rounds there is not enough time for the provers to communicate, but they may communicate over multiple rounds. 

As for the security of this scheme, it is obvious that the hiding property stays satisfied up to the open phase: Every single message $V$ receives is one-time-pad encrypted. As for the binding property, Lunghi \etal. prove that the scheme with a $m$-round sustain phase is $\eps_m$-binding against classical attacks, where $\eps_m$ satisfies $\eps_0 = 2^{-n}$ (this is just the standard Cr{\'e}peau \etal. scheme) and 
$
\eps_m \leq 2^{-n-1} + \sqrt{\eps_{m-1}} 
$
for $m \geq 1$. Thus, even when reading this recursive formula liberally by ignoring the $2^{-n-1}$ term, we obtain 
$$
\eps_m \lesssim \sqrt[2^m]{\eps_0} = 2^{-\frac{n}{2^m}} \, ,
$$
i.e., the error parameter blows up {\em double exponentially} in $m$.%
\footnote{Lunghi \etal. also provide a more complicated recursive formula for $\eps_m$ that is slightly better, but the resulting blow-up is still double exponential.}
In other words, in order to have a non-trivial $\eps_m$ we need that $n$, the size of the strings that are communicated, is {\em exponential} in $m$. This means that Lunghi \etal. can only afford a very small number of rounds. 
For instance, in their implementation where they can manage $n = 512$ (beyond that, the local computation takes too long), asking for an error parameter $\eps_m$ of approximately $2^{-32}$, they can do $m = 4$ rounds.%
\footnote{Note that \cite{LKB+15} mentions $\eps_m \approx 10^{-5} \approx 2^{-16}$, but this is an error, as communicated to us by the authors, and as can easily be verified. Also, \cite{LKB+15} mentions $m = 5$ rounds, but this is because they include the commit round in their counting, and we do not. }
This allows them to keep a commitment alive for $2\,$ms.

\paragraph{\sc Our Results. }
Our main goal is to improve the bound on the binding parameter of the above multi-round scheme. Indeed, our results show that the binding parameter blows up only {\em linearly} in $m$, rather than double exponentially. Explicitly, our results show that (for classical attacks)
$$
\eps_m \leq (m+1)\cdot 2^{-\frac{n}{2}+2} \, .
$$
Using the same $n$ and error parameter as in the implementation of Lunghi \etal., we can now afford approximately $m = 2^{224}$ rounds. Scaling up the $2\,$ms from the Lunghi \etal. experiment for $4$ rounds gives us a time that is in the order of $10^{56}$ years.
We also show tightness of our bound up to a small constant factor (for even $n$).

We use the following strategy to obtain our improved bound on $\eps_m$. We observe that the first sustain round can be understood as committing on the opening information $y_0$ of the actual commitment, using an extended version of the Cr{\'e}peau \etal. scheme that commits to a {\em string} rather than to a bit. Similarly, the second sustain round can be understood as committing on the opening information $y_1$ of that commitment from the first sustain round, etc. Thus, thinking of the $m=1$ version of the scheme, what we have to prove is that if we have two commitment schemes $\cal S$ and ${\cal S}'$, and we modify the opening phase of $\cal S$ in that we first commit to the opening information (using ${\cal S}'$) and then open that commitment, then the resulting commitment scheme is still binding; note that, intuitively, this is what one would indeed expect. Given such a  composition theorem, we can then apply it inductively and conclude security (i.e. the binding property) of the Lunghi \etal. 
multi-round scheme. 

Our main result is such a general composition theorem, which shows that if $\cal S$ and ${\cal S}'$ are respectively $\eps$- and $\delta$-binding (against classical attacks) then the composed scheme is $(\eps+\delta)$-binding (against classical attacks), under some mild assumptions on $\cal S$ and~${\cal S}'$. Hence, the error parameters simply add up; this is what gives us the linear growth. The proof of our composition theorem crucially relies on new definitions of the binding property of 2-prover commitment schemes, which seem to be handier to work with than the $p_0+p_1 \leq 1 + \eps$ definition as for instance used by Lunghi \etal. Our definitions formalize the following intuitive requirement: After the commit phase, even if the provers are dishonest, there should exist some bit $\hat{b}$ such that opening the commitment to any other bit fails (with high probability). We show that one of our new definitions is equivalent to the $p_0+p_1$-definition, while the other one is strictly stronger. Our result 
holds for both definitions, so we not only obtain a better parameter than Lunghi \etal. but also with respect to a stronger definition, and thus we improve the result also in that direction.

One subtle issue is that the extended version of the Cr{\'e}peau \etal. scheme to strings, as it is used in the sustain phase, is not a fully secure string commitment scheme. The reason is that for {\em any} $y$ that may be announced in the opening phase, there exists a string $s$ such that $x + y = a \cdot s$; as such, the provers can commit to some fixed string, and then can still decide to either open the commitment to that string (by running the opening phase honestly), or to open it to a random string that is out of their control (by announcing a random $y$). We deal with this by also introducing a {\em relaxed} version (which we call \em fairly-binding\em ) of the binding property, which captures this limited freedom for the provers, and we show that it is satisfied by the (extended version of the) Cr{\'e}peau \etal. scheme and that our composition theorem holds for this relaxed version;
finally, we observe that the composed fairly-binding string commitment scheme is a  binding {\em bit} commitment scheme when restricting the domain to a bit. 

As such, we feel that our techniques and insights not only give rise to an improved analysis of the Lunghi \etal. multi-round scheme, but they significantly improve our understanding of the security of 2-prover commitment schemes, and as such are likely to find further applications.

\vspace{-0.5ex}
\paragraph{\sc Open Problems. }
Our work gives rise to a list of interesting and challenging open problems. For instance, our composition theorem only applies to pairs ${\cal S},{\cal S}'$ of commitment schemes of a certain restricted form, e.g., only one prover should be involved in the commit phase (as it is the case in the Cr{\'e}peau \etal. scheme). Our proof crucially relies on this, but there seems to be no fundamental reason for such a restriction. Thus, we wonder if it is possible to generalize our composition theorem to a larger class of pairs of schemes, or, ultimately, to {\em all} pairs of schemes (that ``fit together''). 

In another direction, some of our observations and results generalize immediately to the quantum setting, where the two dishonest provers are allowed to compute their messages by performing measurements on an entangled quantum state, but in particular our main result, the composition theorem, does not generalize. Also here, there seems to be no fundamental reason, and thus, generalizing our composition theorem to the quantum 
setting is an interesting open problem. 
Finally, in order to obtain security of the Lunghi \etal. multi-round scheme against quantum attacks, beyond a quantum version of the composition theorem, one also needs to prove security against quantum attacks of the (extended version of the) original Cr{\'e}peau \etal. scheme as a (fairly-binding) {\em string} commitment scheme.

\vspace{-0.5ex}
\paragraph{\sc Concurrent Work. } In independent and concurrent work, Chakraborty, Chailloux and Leverrier \cite{CCL15} showed (almost) the same linear bound for the Lunghi \etal. scheme, but with respect to the original\,---\,and thus weaker\,---\,notion of security. Their approach is more direct and tailored to the specific scheme; our approach is more abstract and provides more insight, and our result applies much more generally.

\section{Preliminaries}

\subsection{Basic Notation}

\paragraph{\sc Probability Distributions. }
For the purpose of this work, a {\em (probability) distribution} is a 
function $p: {\cal X} \rightarrow [0,1]$, $x \mapsto p(x)$, where $\cal X$ 
is a finite non-empty set, with the property that $\sum_{x \in \cal X} p(x) = 1$. 
For specific choices $x_\circ \in \cal X$, we tend to write $p(x\!=\!x_\circ)$ instead of $p(x_\circ)$. 
For any subset $\Lambda \subset \cal X$, called an {\em event}, the probability $p(\Lambda)$ is naturally 
defined as $p(\Lambda) = \sum_{x \in \Lambda} p(x)$, and it holds that
\begin{equation}
\label{eq:simple_pr_sum}
p(\Lambda) + p(\Gamma) = p(\Lambda \cup \Gamma) + p(\Lambda \cap \Gamma) 
\leq 1 + p(\Lambda \cap \Gamma)
\end{equation}
for all $\Lambda,\Gamma \subset \cal X$, and, more generally, that
\begin{equation}
\label{eq:pr_sum}
\sum_{i=1}^k p(\Lambda_i) 
\leq p(\Lambda_1 \cup \ldots \cup \Lambda_k) + \sum_{i < j} p(\Lambda_i \cap \Lambda_j) 
\leq 1 + \sum_{i < j} p(\Lambda_i \cap \Lambda_j) 
\end{equation}
for all $\Lambda_1,\ldots,\Lambda_k \subset \cal X$. 
For a distribution $p: 
{\cal X} \times {\cal Y} \rightarrow \R$ on two (or more) variables, probabilities like 
$p(x\!=\!y)$, $p(x\!=\!f(y))$, $p(x\!\neq\!y)$ etc. are naturally 
understood as
$$
p(x=y) = p\bigl(\Set{(x,y) \in {\cal X} \times {\cal Y}}{x = y}\bigr) = 
\sum_{x \in {\cal X}, y \in {\cal Y} \atop \text{s.t. } x = y} p(x,y)
$$
etc., and the {\em marginals} $p(x)$ and $p(y)$ are given by 
$p(x) = \sum_y p(x,y)$ and $p(y) = \sum_x p(x,y)$, respectively. 
Vice versa, given two distributions $p(x)$ and $p(y)$, we say that a distribution $p(x,y)$ on two variables is a {\em consistent joint distribution} if the two marginals of $p(x,y)$ coincide with $p(x)$ and $p(y)$, respectively. We will make use of the following property on the existence of a consistent joint distribution that maximizes the probability that $x = y$; the proof is given in the appendix. 

\begin{lemma}
  \label{lm:eps_dist}
Let $p(x)$ and $p(y)$ be two distributions on a common set $\cal X$. Then there exists a consistent joint distribution $p(x,y)$ such that $p(x=y=x_\circ) = \min\set{p(x\!=\!x_\circ),p(y\!=\!x_\circ)}$ for all choices of $x_\circ \in \cal X$. Additionally, $p(x,y)$ satisfies $p(x,y|x\neq y) = p(x|x\neq y) \cdot p(y|x\neq y)$. 
\end{lemma}

\paragraph{\sc Protocols. }
In this work, we will consider 3-party (interactive) {\em protocols}, where the parties are named $P$, $Q$ and $V$ (the two ``provers'' and the ``verifier''). Such a protocol $\prot_{PQV}$ consists of a triple $(\prot_P,\prot_Q,\prot_V)$ of $L$-round {\em interactive algorithms} for some $L \in \N$. Each interactive algorithm takes an input, and for every round $\ell \leq L$ computes the messages to be sent to the other algorithms/parties in that round as deterministic functions of its input, the messages received in the previous rounds, and the local randomness. In the same way, the algorithms produce their respective outputs after the last round. We write
$$
(out_P\|out_Q\|out_V)\leftarrow \bigl(\prot_P(in_P)\|\prot_Q(in_Q)\|\prot_V(in_V)\bigr)
$$
to denote the execution of the protocol $\prot_{PQV}$ on the respective inputs $in_P, in_Q$ and $in_V$, and that the respective outputs $out_P,out_Q$ and $out_V$ are produced. Clearly, for any protocol $\prot_{PQV}$ and any input $in_P, in_Q, in_V$, the probability distribution $p(out_P,out_Q,out_V)$ of the output is naturally well defined. 

If we want to make the local randomness explicit, we write $\prot_P[\xi_P](in_P)$ etc., and understand that $\xi_P$ is correctly sampled\,---\,without loss of generality, we may assume it to be a uniformly random bit string of sufficient length. Furthermore, we write $\prot_P[\xi_{PQ}](in_P)$ and $\prot_Q[\xi_{PQ}](in_Q)$ to express that $\prot_P$ and $\prot_Q$ use {\em the same} randomness, in which case we speak of {\em joint randomness}.   

We can {\em compose} two interactive algorithms $\prot_P$ and $\prot'_P$ in the obvious way, by applying $\prot'_P$ to the output of $\prot_P$. The resulting interactive algorithm is denoted as $\prot'_P \circ \prot_P$. Composing the respective algorithms of two protocols $\prot_{PQV} = (\prot_P,\prot_Q,\prot_V)$ and $\prot'_{PQV} = (\prot'_P,\prot'_Q,\prot'_V)$ results in the composed protocol $\prot'_{PQV} \circ \prot_{PQV}$. If $\prot_P$ is a {\em non-interactive} algorithm, then $\prot'_{PQV} \circ \prot_{P}$ is naturally understood as the protocol $\prot'_{PQV} \circ \prot_{P} = (\prot'_P \circ \prot_{P},\prot'_Q,\prot'_V)$, and similarly $\prot'_{PQV} \circ \prot_{QV}$ in case $\prot_{QV}$ is a protocol among $Q$ and $V$ only. 

\subsection{2-Prover Commitment Schemes}

\begin{definition}\label{def:com}
A {\em 2-prover (string) commitment scheme} $\cal S$ consists of two interactive protocols $\com_{PQV} = (\com_P,\com_Q,\com_V)$ and $\open_{PQV} = (\open_P,\open_Q,\open_V)$ between the two provers $P$ and $Q$ and the verifier $V$, with the following syntactics. 
The {\em commit protocol} $\com_{PQV}$ uses joint randomness $\rndc_{PQ}$ for $P$ and $Q$ and takes a string $s \in \set{0,1}^n$ as input for $P$ and $Q$ (and independent randomness and no input for~$V$), and it outputs a {\em commitment} $com$ to $V$ and some state information to $P$ and $Q$: 
$$
(state_P\|state_Q\|c)\leftarrow\bigl(\com_P[\rndc_{PQ}](s)\|\com_Q[\rndc_{PQ}](s)\|\com_V(\emptyset)\bigr) \, .
$$
The {\em opening protocol} $\open_{PQV}$ uses joint randomness $\rndo_{PQ}$ and outputs a string or a rejection symbol to $V$, and nothing to $P$ and $Q$: 
$$
(\emptyset\|\emptyset\|s) \leftarrow \bigl(\open_P[\rndo_{PQ}](state_P)\|\open_Q[\rndo_{PQ}](state_Q)\|\open_V(c)\bigr)
$$
with $s \in \set{0,1}^n \cup \set{\bot}$. 
The set $\set{0,1}^n$ is called the {\em domain} of $\cal S$; if $n = 1$ then we refer to $\cal S$ as a {\em bit} commitment scheme instead, and we tend to use $b$ rather than $s$ to denote the committed bit. 
\end{definition}

\begin{remark}
\label{rem:entire_com}
By convention, we assume throughout the paper that the commitment $c$ output by $V$ equals the {\em communication} that
takes place between $V$ and the provers during the commit phase. This is without loss of generality since, in general, $c$ is computed as a (possibly randomized) function of the communication, which $V$ just as well can apply in the opening phase. 
\end{remark}
\begin{remark}
Note that we specify that $P$ and $Q$ use {\em fresh} joint randomness $\rndo_{PQ}$ in the opening phase, and, if necessary, the randomness $\rndc_{PQ}$ from the commit phase can be ``handed over'' to the opening phase via $state_P$ and $state_Q$; this will be convenient later on. Alternatively, one could declare that $P$ and $Q$ {\em re-use} the joint randomness from the commit phase.  
\end{remark}
Whenever we refer to such a 2-prover commitment scheme, we take it as understood that the scheme is complete and hiding, as defined below, for ``small'' values of $\gamma$ and $\delta$. Since our focus will be on the binding property, we typically do not make the parameters $\eta$ and $\delta$ explicit. 

\begin{definition}
A 2-prover commitment scheme is {\em $\gamma$-complete} if in an honest execution $V$'s output $s$ of $\open_{PQV}$ equals $P$ and $Q$'s input $s$ to $\com_{PQV}$ except with probability $\eta$, for any choice of $P$ and $Q$'s input $s \in \bin^n$.
\end{definition}

The standard definition for the hiding property is as follows:
\begin{definition}
\label{def:hiding_standard}
A 2-prover commitment scheme is {\em $\delta$-hiding} if for any commit
strategy $\dcom_V$ and any two strings $s_0$ and $s_1$, the distribution of the
commitments $c_0$, $c_1$, produced as
$$
  (state_P\|state_Q\|c_b)\leftarrow(\com_P[\rndc_{PQ}](s_b)\|\com_Q[\rndc_{PQ}](s_b)\|\dcom_V(\emptyset))\text{,} b = 0,1
$$
have statistical distance at most $\delta$. A $0$-hiding scheme is also called
{\em perfectly hiding}. 
\end{definition}

Defining the binding property is more subtle. First, note that an attack against the binding property consists of an ``allowed'' commit strategy $\dcom_{PQ} = (\dcom_P,\dcom_Q)$ and an ``allowed'' opening strategy $\dopen_{PQ} = (\dopen_P,\dopen_Q)$ for $P$ and $Q$. 
Any such attack fixes $p(s)$, the distribution of $s \in \set{0,1}^n \cup \set{\bot}$ that is output by $V$ after the opening phase, in the obvious way. 

What exactly ``allowed'' means may depend on the scheme and needs to be specified. 
Typically, in the 2-prover setting, we only allow strategies $\dcom_{PQ}$ and $\dopen_{PQ}$ with {\em no communication} at all between the two provers, but we may also be more liberal and allow some {\em well-controlled} communication, as in the Lunghi \etal . multi-round scheme.
Furthermore, in this work, we focus on {\em classical} attacks, where $\dcom_{P},\dcom_{Q},\dopen_{P}$ and $\dopen_{Q}$ are classical interactive algorithms as specified in the previous section, with access to joint randomness. But one could also consider {\em quantum} attacks, in which the provers can perform measurements on an entangled quantum state.
Our main result holds for classical attacks only, and so the unfamiliar reader can safely ignore the possibility of quantum attacks, but some of our insights also apply to quantum attacks. 

A somewhat accepted definition for the binding property of a 2-prover {\em bit} commitment scheme, as it is for instance used in~\cite{CSST11,LKB+15,FF15} (up to the factor $2$ in the error parameter), is as follows. Here, we assume it has been specified which attacks are {\em allowed}, e.g., those where $P$ and $Q$ do not communicate during the course of the scheme. 

\begin{definition}\label{def:p0+p1}
A 2-prover bit commitment scheme is {\em $\eps$-binding in the sense of $p_0 + p_1 \leq 1+2\eps$} if for every allowed commit strategy $\dcom_{PQ}$, and for every pair of allowed opening strategies $\dopen^0_{PQ}$ and \smash{$\dopen^1_{PQ}$}, which fix distributions $p(b_0)$ and $p(b_1)$ for $V$'s respective outputs, it holds that 
$$
p(b_0\!=\!0) + p(b_1\!=\!1) \leq 1+2\eps \, .
$$ 
\end{definition}
In the literature (see~e.g.~\cite{CSST11} or~\cite{LKB+15}), the two probabilities $p(b_0\!=\!0)$ and $p(b_1\!=\!1)$ above are usually referred to as $p_0$ and $p_1$, respectively. 

\subsection{The $\CHSH^n$ Scheme}

Our main example is the bit commitment scheme by Cr{\'e}peau \etal.~\cite{CSST11} we mentioned in the introduction, and which works as follows. The commit phase $\com_{PQV}$ instructs $V$ to sample and send to $P$ a uniformly random $a \in \set{0,1}^n$, and it instructs $P$ to return $x:= r + a \cdot b$ to $V$, where $r$ is the joint randomness, uniformly distributed in $\bin^n$, and $b$ is the bit to commit to, and the opening phase $\open_{PQV}$ instructs $Q$ to send $y := r$ to $V$, and $V$ outputs the (smaller) bit $b$ that satisfies $x + y = a \cdot b$, and $b := \bot$ in case no such bit exists.
Note that the provers in this scheme use the same randomness in the
commit and opening phase; thus, formally, $Q$ needs to output the shared
randomness $r\leftarrow\rndc_{PQ}$ as $state_Q$. The opening phase uses no fresh
randomness.

It is easy to see that this scheme is $2^{-n}$-complete and perfectly hiding (completeness fails in case $a = 0$). 
For {\em classical} provers that do not communicate at all, the scheme is $2^{-n-1}$-binding in the sense of $p_0 + p_1 \leq 1+2^{-n}$, i.e. according to Definition~\ref{def:p0+p1}. As for {\em quantum} provers, Cr{\'e}peau \etal. showed that the scheme is $2^{-n/2}$-binding; this was recently minorly improved to $2^{-(n+1)/2}$ by Sikora, Chailloux and Kerenidis~\cite{SCK14}. 

We also want to consider an extended version of the scheme, where the bit $b$ is replaced by a string $s \in \set{0,1}^n$ in the obvious way (where the multiplication $a \cdot s$ is then understood in a suitable finite~field), and we want to appreciate this extension as a 2-prover {\em string} commitment scheme. 
However, it is a priori not clear what is a suitable definition for the binding property, especially because for this particular scheme, the dishonest provers can always honestly commit to a string $s$, and can then decide to correctly open the commitment to $s$ by announcing $y := r$, or open to a {\em random} string by announcing a randomly chosen $y$\,---\,any $y$ satisfies $x + y = a \cdot s$ for {\em some} $s$ (unless $a = 0$, which almost never happens).%
\footnote{This could easily be prevented by requiring $Q$ to announce $s$ (rather than letting $V$ compute it), but we want the information announced during the opening phase to fit into the domain of the commitment scheme. }

Due to its close relation to the CHSH game~\cite{CHSH69}, in particular to the arbitrary-finite-field version considered in~\cite{BS15}, we will refer to this {\em string} commitment scheme as $\CHSH^n$.

\section{On the Binding Property of 2-Prover Commitment Schemes}

We introduce new definitions for the binding property of 2-prover commitment schemes. In the case of {\em bit} commitment schemes, they imply Definition~\ref{def:p0+p1}, as we will show. Although not necessarily simpler, we feel that our definitions are closer to the intuition of what is expected from a commitment scheme, and as such easier to work with. Indeed, the proofs of our composition results are heavily based on our new definitions. 
Also, our new notions are more flexible in terms of tweaking it; for instance, we modify them to obtain a {\em relaxed} notion for the binding property, which captures the binding property that is satisfied by the string commitment scheme $\CHSH^n$. 

Throughout this section, when quantifying over attacks against (the binding property of) a scheme, it is always understood that there is a notion of {\em allowed} attacks for that scheme (e.g., all attacks for which $P$ and $Q$ do not communicate), and that the quantification is over all such allowed attacks. Also, even though our focus is on classical attacks,
Proposition~\ref{prop:fairlyweaktostrong} and Theorem~\ref{thm:equiv_weak} also
apply to quantum attacks.

\subsection{Defining The Binding Property}\label{sec:defs}

Intuitively, we say that a scheme is binding if after the commit phase there
exists a string $\hat{s}$ so that no matter what the provers do in the opening
phase, the verifier will output either $s = \hat{s}$ or $s = \bot$ (except with
small probability). We consider two definitions of the binding property which
interpret this intuitive requirement in two different ways. In the first
definition, which we introduce in this section, $\hat{s}$ is a function of the
provers' (combined) view immediately after the commit phase. In the second one,
which we introduce in Section~\ref{sec:defs_weak}, $\hat{s}$ is specified by its distribution only. Both of these definitions admit a composition theorem.

\begin{definition}[Binding property]
  \label{def:bind}
A 2-prover commitment scheme $\mathcal S$ is \em $\eps$-binding \em if for every
commit strategy $\dcom_{PQ}[\drndc_{PQ}]$ there exists a function $\hat{s}(\drndc_{PQ},c)$
of the joint randomness \smash{$\drndc_{PQ}$} and the commitment%
\footnote{Recall that by convention (Remark~\ref{rem:entire_com}), $c$ equals the communication between $V$ and the provers during the commit phase. }
$c$ such that
for every opening strategy $\dopen_{PQ}$ it holds that
$p(s \neq \hat{s}(\drndc_{PQ},c) \land s\neq\bot) \leq \eps$. In short: 
\begin{equation}
\forall\,\dcom_{PQ} \; \exists \, \hat{s}(\drndc_{PQ},c) \; \forall \, \dopen_{PQ} : p(s \neq \hat{s}\land s\neq\bot) \leq \eps \, .
\end{equation}
\end{definition}
The string commitment scheme $\CHSH^n$ does {\em not} satisfy this definition (the bit commitment version does, as we will show): After the commit phase, the provers can still decide to open the commitment to a {\em fixed} string, chosen before the commit phase, or to a {\em random} string that is out of their control. 
We capture this by the following relaxed version of the binding property: We allow $V$'s output $s$ to be different from $\hat{s}$ and $\bot$, but in this case the provers should have little control over $s$: For any fixed {\em target string} $s_\circ$, it should be unlikely that $s = s_\circ$. Formally, this is captured as follows; we will show in Section~\ref{sec:CSST} that $\CHSH^n$ is fairly-binding in this sense.

\begin{definition}[Fairly binding property]
  \label{def:fairly}
A 2-prover commitment scheme $\mathcal S$ is $\eps$-\em fairly-binding \em if for every commit strategy $\dcom_{PQ}[\drndc_{PQ}]$ there exists a function
$\hat{s}(\drndc_{PQ},c)$ such that for every opening strategy $\dopen_{PQ}[\drndo_{PQ}]$ 
and all functions $s_\circ(\drndc_{PQ}, \drndo_{PQ})$ it holds that
$p(s \neq \hat{s}(\drndc_{PQ},c) \,\land\, s = s_\circ(\drndc_{PQ}, \drndo_{PQ})) \leq \eps$. In short: 
\begin{equation}
\label{eq:fairly_bind}
\forall\, \dcom_{PQ} \; \exists \, \hat{s}(\drndc_{PQ},c) \; \forall \, \dopen_{PQ} \; \forall \, s_\circ(\drndc_{PQ},\drndo_{PQ}) : p(s \neq \hat{s} \,\land\, s = s_\circ) \leq \eps \, .
\end{equation}
\end{definition}

\begin{remark}\label{rem:det}
By means of standard techniques, one can easily show that
it is sufficient for the (fairly) binding property to consider
\em deterministic \em provers. In this case, $\hat{s}$ is a function of $c$
only, and, in the case of fairly-binding, $s_\circ$ runs over all {\em fixed}
strings.
\end{remark}

\begin{remark}
\label{rem:bit_equiv}
Clearly, the binding property implies the fairly binding property. Furthermore, in the case of {\em bit} commitment schemes it obviously holds that $p(b \neq \hat{b} \land b \neq \bot) = p(b \neq \hat{b} \land b = 0) + p(b \neq \hat{b} \land b = 1)$, and thus the fairly-binding property implies the binding property with a factor-2 loss in the parameter. Furthermore, every fairly-binding {\em string} commitment scheme gives rise to a binding {\em bit} commitment scheme in a natural way, as shown by the following proposition. 
\end{remark}

\begin{proposition}\label{prop:fairlytostrong}
Let $\mathcal S$ be a $\eps$-fairly-binding string commitment scheme.
Fix any two distinct strings $s_0, s_1 \in \bin^n$ and consider the
bit-commitment scheme $\mathcal S'$ obtained as follows. To commit to
$b \in \bin$, the provers commit to $s_b$ using $\cal S$, and in the opening
phase $V$ checks if $s = s_b$ for some bit $b \in \bin$ and outputs this bit if
it exists and else outputs $b = \bot$. Then, $\mathcal S'$ is a
$2\eps$-binding bit commitment scheme.
\end{proposition}

\begin{proof}
Fix some commit strategy $\dcom_{PQ}$ for $\mathcal S'$ and note that it
can also be used to attack $\mathcal S$. Thus, there exists a function
$\hat{s}(\drndc_{PQ},c)$ as in Definition \ref{def:fairly}. We define
$$
  \hat{b}(\drndc_{PQ},c) =
  \begin{cases}
    0 &\text{ if }\hat{s}(\drndc_{PQ},c) = s_0\\
    1 &\text{ otherwise}
  \end{cases}
$$
Now fix an opening strategy $\dopen _{PQ}$ for $\mathcal S'$, which again is
also a strategy against $\mathcal S$. Thus, we have
$p(\hat{s} \neq s = s_\circ)\leq \eps$ for any $s_\circ$
(and in particular $s_\circ = s_0$ or $s_1$).
This gives us
\begin{align*}
  p(\bhat \neq b \neq \bot)
  &= p(\bhat = 1\land b = 0) + p(\bhat = 0 \land b = 1)\\
  &= p(\hat{s} \neq s_0\land s = s_0) + p(\hat{s} = s_0 \land s = s_1)\\
  &\leq p(\hat{s} \neq s_0\land s = s_0) + p(\hat{s} \neq s_1 \land s = s_1) \, ,\\
  &\leq 2\eps
\end{align*}
and thus $\mathcal S'$ is a $2\eps$-binding bit-commitment scheme.\qed
\end{proof}

\begin{remark}
\label{rem:genfairlytostrong}
The proof of Proposition~\ref{prop:fairlytostrong} generalizes in a
straightforward way to $k$-bit string commitment schemes: Given a
$\eps$-fairly-binding $n$-bit string commitment scheme ${\cal S}$, for $k < n$,
we define a $k$-bit string commitment scheme ${\cal S}_k$ as follows: To commit
to a $k$-bit string, the provers pad the string with $n-k$ zeros and then
commit to the padded string using $\cal S$. In the opening phase, the verifier
outputs the first $k$ bits of $s$ if the remaining bits in $s$ are all zeros,
and $\bot$ otherwise. Then, ${\cal S}'$ is $2^k\eps$-binding.
\end{remark}

\subsection{The Weak Binding Property}
\label{sec:defs_weak}

Here, we introduce yet another definition for the binding property. It is similar in spirit to Definition~\ref{def:bind}, but weaker. One advantage of this weaker notion is that it is also meaningful when considering quantum attacks, whereas Definition~\ref{def:bind} is not. 
In the subsequent section, we will see that for {\em bit} commitment schemes, this weaker notion of the binding property is equivalent to Definition~\ref{def:p0+p1}. 

\begin{definition}[Weak binding property]
  \label{def:weak}
A 2-prover commitment scheme $\mathcal S$ is \em $\eps$-weak-binding \em if for
all commit strategies $\dcom _{PQ}$ there exists a distribution $p(\hat{s} )$
such that for every opening strategy $\dopen _{PQ}$ (which then fixes the
distribution $p(s)$ of $V$'s output $s$) there is a consistent joint
distribution $p(\hat{s}, s)$ such that
$p(s \neq \hat{s} \land s\neq\bot) \leq \eps$. In short: 
\begin{equation}
\forall\,\dcom_{PQ} \; \exists \, p(\hat{s}) \; \forall \, \dopen_{PQ} \; \exists \, p(\hat{s},s) : p(s \neq \hat{s}\land s\neq\bot) \leq \eps \, .
\end{equation}
\end{definition}
We also consider a related, i.e., ``fairly'', version of this binding property,
similar to Definition~\ref{def:fairly}.
\begin{definition}[Fairly weak binding property]
  \label{def:fairly_weak}
A 2-prover commitment scheme $\mathcal S$ is $\eps$-\em fairly-weak-binding \em if for all commit strategies $\dcom _{PQ}$ there exists a distribution $p(\hat{s} )$ such that for every 
  opening strategy $\dopen _{PQ}$ (which then fixes the distribution $p(s)$ of $V$'s output $s$) there
  is a consistent joint distribution $p(\hat{s}, s)$ so that for all $s_\circ\in\bin^n$ it holds that
  $p(s \neq \hat{s} \land s = s_\circ) \leq \eps$. In short: 
\begin{equation}
\label{eq:fairly_weak_bind}
\forall\, \dcom_{PQ} \; \exists \, p(\hat{s}) \; \forall \, \dopen_{PQ} \; \exists \, p(\hat{s},s) \; \forall \, s_\circ  : p(s \neq \hat{s} \,\land\, s = s_\circ) \leq \eps \, .
\end{equation}
\end{definition}

\begin{remark}
\label{rem:weak_and_strong}
Remarks~\ref{rem:det} and \ref{rem:bit_equiv} also hold for the
weak binding properties. Furthermore, it is easy to see that the binding and fairly-binding
properties imply their weak counterparts.
\end{remark}

\begin{proposition}\label{prop:fairlyweaktostrong}
Let $\mathcal S$ be a $\eps$-fairly-weak-binding string commitment scheme
and define ${\cal S}'$ as in Proposition~\ref{prop:fairlytostrong}.
Then, $\mathcal S'$ is a $2\eps$-weak-binding bit commitment scheme.
\end{proposition}
\begin{proof}
The proof of Proposition~\ref{prop:fairlytostrong} can be easily adapted: Let
$p(\hat{s})$ be as required by Definition~\ref{def:fairly_weak}. We define
$p(\hat{b})$ by taking the marginal of $p(\hat{s}, \hat{b})$ where
$\hat{b} = 0$ if $\hat{s} = s_0$, and $\hat{b} = 1$ otherwise. An opening
strategy $\dopen_{PQ}$ for ${\cal S}'$ can also be viewed as a strategy for
${\cal S}$. As such, there is a joint distribution $p(\hat{s}, s)$ as required
by Definition~\ref{def:weak} which we can extend to $p(\hat{s}, s, b)$ by
setting $b = 0$ if $s = s_0$, $b = 1$ if $s = s_1$ and $b=\bot$ otherwise. We
define
$p(\hat{b}, b) := \sum_{\hat{s}, s} p(\hat{b}, \hat{s})\cdot p(s, b|\hat{s})$.
As in the proof of Proposition~\ref{prop:fairlytostrong}, one can easily check
that $p(\hat{b}\neq b\neq\bot) \leq 2\eps$ holds.
\end{proof}

\subsection{Relations Between The Definitions}
\label{sec:equiv}

Here, we show that in case of {\em bit} commitment schemes, the weak binding property as introduced in Definition~\ref{def:weak} above is actually {\em equivalent} to the $(p_0+p_1)$-definition. Even though our focus is on classical attacks, the proof immediately carries over to quantum attacks as well. 
\begin{theorem}\label{thm:equiv_weak}
A 2-prover bit-commitment scheme is $\eps$-binding in the sense of
$p_0+p_1\leq 1+2\eps$ if and only if it is $\eps$-weak-binding. 
\end{theorem}
\begin{proof}
First, consider a scheme that is $\eps$-binding according to Definition~\ref{def:p0+p1}. Fix a commit strategy $\dcom _{PQ}$ and opening strategies
$\dopen _{PQ}^0$ and $\dopen _{PQ}^1$ so that $p_0 = p(b_0=0)$ and
$p_1 = p(b_1 = 1)$ are maximized, where $b_i \in \{ 0, 1, \bot \}$ is $V$'s
output when the dishonest provers use opening strategy $\dopen_{PQ}^i$. Let
$p_0 + p_1 = 1+2\eps'$. Since the scheme is $\eps$-binding, we have
$\eps'\leq \eps$. We define the distribution \smash{$p(\bhat)$} as
$p(\bhat = 0) := p_0-\eps'$ and $p(\bhat = 1) := p_1 - \eps'$. To see that this
is indeed a probability distribution, note that $p_0, p_1 \geq 2\eps'$
(otherwise, we would have $p_0 > 1$ or $p_1 > 1$) and that
$p(\bhat = 0) + p(\bhat = 1) = p_0 + p_1 - 2\eps' = 1$. Now we consider an
arbitrary opening strategy $\dopen _{PQ}$ which fixes a distribution $p(b)$. By
definition of $p_0$ and $p_1$, we have $p(b=i) \leq p_i$ and thus \smash{$p(b=i)\leq p(\bhat = i) + \eps'\leq p(\bhat = i) + \eps$}.
By Lemma \ref{lm:eps_dist}, there exists a consistent joint distribution
$p(\bhat, b)$ with the property that
$p(\bhat = b = i) = \min\set{p(b=i),p(\bhat = i)}$.
We wish to bound $p(\hat{b}\neq b\land b \neq \bot)
= p(\hat{b}=0\land b = 1) + p(\hat{b}=1\land b=0)$.
For $i\in\{0,1\}$, it holds that
\begin{align*}
  p(\hat{b} = 1-i \land b = i)
  &= p(b = i) - p(\hat{b} = b = i)\\
  &= p(b = i) - \min\set{p(\hat{b} = i), p(b=i)}\\
  &= \max\set{0, p(b=i) - p(\hat{b}=i)}\\
  &\leq \eps
\end{align*}
and furthermore, there is at most \em one \em $i\in\{0,1\}$ such that
$p(b=i) > p(\hat{b}=i)$, for if $p(b=i) > p(\hat{b}=i)$ for both $i=0$ and
$i=1$, then $p(b=0)+p(b=1) > p(\hat{b}=0) + p(\hat{b}=1) = 1$ which is a
contradiction.
Thus, we have $p(\hat{b}\neq b\land b\neq\bot) \leq \eps$.
This proves one direction of our claim.

For the other direction, consider a scheme that is $\eps$-binding. Fix
$\dcom _{PQ}$ and let $p(\bhat)$ be a distribution such that for every opening
strategy $\dopen _{PQ}$, there is a joint distribution \smash{$p(\bhat, b)$}
with $p(\bhat\neq b\neq\bot)\leq \eps$. Now consider two opening strategies
$\dopen _{PQ}^0$ and $\dopen _{PQ}^1$ which give distributions $p(b_0)$ and
$p(b_1)$. We need to bound $p(b_0 = 0) + p(b_1 = 1)$. There is a joint
distribution \smash{$p(\bhat, b_0)$} such that $p(\bhat\neq b_0\neq \bot)\leq\eps$
and likewise for $b_1$. Thus,
\begin{align*}
  p(b_0 = 0) + p(b_1 = 1)
  &= p(\bhat = 0, b_0 = 0) + p(\bhat = 1, b_0 = 0)
    + p(\bhat = 0, b_1 = 1) + p(\bhat = 1, b_1 = 1)\\
  &\leq p(\bhat = 0) + p(\bhat = 1) + p(\bhat \neq b_0 \neq \bot)
    + p(\bhat \neq b_1 \neq \bot)\\
  &\leq 1 + 2\eps
\end{align*}
which proves the other direction.\qed
\end{proof}

\begin{remark}
\label{rem:sep}
By Remark~\ref{rem:weak_and_strong}, it follows that Definition~\ref{def:bind}
also implies the $p_0+p_1$-definition. In fact, Definition~\ref{def:bind} is
\em strictly \em stronger (and hence, also strictly stronger than the
weak-binding definition). Consider the following
(artificial and very non-complete) scheme: In the commit phase, $V$ chooses a
uniformly random bit and sends it to the provers, and then accepts everything
or rejects everything during the opening phase, depending on that bit. Then,
$p_0 + p_1 = 1$, yet a commitment can be opened to $1-\hat{b}$ (no matter how
$\hat{b}$ is defined) with probability $\frac12$.

Since a non-complete separation example may not be fully satisfying, we note
that it can be converted into a complete (but even more artificial) scheme.
Fix a ``good'' (i.e., complete, hiding and binding with low parameters) scheme
and call our example scheme above the ``bad'' scheme. We define a
\em combined \em scheme as follows: At the start, the first prover can request
either the ``good'' or ``bad'' scheme to be used. The honest prover is
instructed to choose the former, guaranteeing completeness. The dishonest
prover may choose the latter, so the combined scheme inherits the binding
properties of the ``bad'' scheme: It is binding according to the
$(p_0+p_1)$-definition, but not according to Definition~\ref{def:bind}.
\end{remark}

\subsection{Security of $\CHSH^n$}\label{sec:CSST}

In this section, we show that $\CHSH^n$ is a fairly-binding string commitment
scheme.%
\footnote{It is understood that the allowed attacks against $\CHSH^n$ are those where the provers do not communicate.} 
To this end, we introduce yet another version of the binding property and show
that $\CHSH^n$ satisfies this property. Then we show that this version of the
binding property implies the fairly-binding property (up to some loss in the
parameter, and some mild restrictions on the scheme). 

This new binding property is based on the intuition that it should not be possible to open a commitment to two different values {\em simultaneously} (except with small probability). For this, we observe that (for {\em classical} attacks), when considering a commit strategy $\dcom_{PQ}$, as well as {\em two} opening strategies $\dopen_{PQ}$ and $\dopen'_{PQ}$, we can run both opening strategies {\em simultaneously} on the produced commitment with two (independent) copies of $\open_V$, by applying $\dopen_{PQ}$ and $\dopen'_{PQ}$ to two copies of the respective internal states of $P$ and $Q$). This gives rise to a {\em joint} distribution $p(s,s')$ of the respective outputs $s$ and $s'$ of the two copies of $\open_V$. 

\begin{definition}[Simultaneous opening]
  \label{def:sim_open}
  A 2-prover commitment scheme $\mathcal S$ is \em $\eps$-fairly-binding in the sense of
simultaneous opening\em \footnote{We use ``fairly'' here to distinguish the notion from 
a ``non-fairly'' version with
$p(\bot\neq s\neq s'\neq\bot)\leq\eps$; however, we do not consider this latter version any further here.}
if for all $\dcom _{PQ}$, all pairs of opening strategies
  $\dopen _{PQ}$ and $\dopen _{PQ}'$, and all pairs $s_\circ,s_\circ '$ of
  distinct strings, we have $p(s = s_\circ \land s' = s_\circ')\leq \eps$. 
\end{definition}

\begin{remark}
\label{rem:sim_open_det}
Also for this notion of fairly-binding, it is sufficient to consider \em deterministic \em strategies, as can easily be seen. 
\end{remark}

\begin{proposition}
  \label{prop:csst_sim_open}
The string commitment scheme $\CHSH^n$ is $2^{-n}$-fairly-binding in the sense of simultaneous opening.
\end{proposition}

\begin{proof}
By Remark~\ref{rem:sim_open_det}, it suffices to consider deterministic attack
strategies. Fix a deterministic strategy $\dcom_{PQ}$ and two deterministic
opening strategies $\dopen_{PQ}$ and $\dopen_{PQ}'$. The strategy $\dcom_{PQ}$
specifies $P$'s output $x$ as a function $f(a)$ of the verifier's message $a$.
The opening strategies are described by constants $y$ and $y'$. By definition
of $\CHSH^n$, $s = s_\circ$ implies $f(a) + y = a\cdot s_\circ$ and likewise,
$s' = s_\circ'$ implies $f(a) +y'  = a\cdot s_\circ'$. 
Therefore, $s = s_\circ \land s' = s_\circ'$ implies $a = (y-y')/(s_\circ-s'_\circ)$.
It thus holds that 
$
  p(s = s_\circ \land s' = s_\circ')
  \leq p\bigl(a = (y-y')/(s_\circ-s'_\circ)\bigr)
  \leq \frac{1}{2^n}
$, which proves our claim.
\qed
\end{proof}

\begin{remark}\label{rem:comparison}
It follows directly from (\ref{eq:simple_pr_sum}) that every {\em bit} commitment scheme that is $\eps$-fairly-binding in the sense of simultaneous opening (against classical attacks) is  $\eps/2$-binding in the sense of $p_0 + p_1 \leq 1 + \eps$ (and thus also according to Definitions~\ref{def:weak}). The converse is
not true though: The schemes from Remark~\ref{rem:sep} again serve as
counterexamples.
\end{remark}

\begin{theorem}
  \label{thm:sim_open2fairly}
Let ${\cal S} = (\com_{PQV}, \open_{PQV})$ be a 2-prover commitment scheme. If
$\cal S$ is $\eps$-fairly-binding in the sense of simultaneous opening and
$\open_V$ is deterministic, then $\cal S$ is $2\sqrt{\eps}$-fairly-binding. 
\end{theorem}
\begin{proof}
By Remark~\ref{rem:det}, it suffices to consider deterministic strategies for
the provers. We fix some deterministic commit strategy $\dcom_{PQ}$ and an
enumeration $\{ \dopen_{PQ}^i \}_{i=1}^N$ of all deterministic opening
strategies. Since we assume that $\open_V$ is deterministic, for any fixed
opening strategy for the provers, 
the verifier's output $s$ is a \em function \em of the commitment $c$. Thus, for each opening
strategy $\dopen_{PQ}^i$ there is a function $f_i$ such that the verifier's
output is $s = f_i(c)$. We will now define the function $\hat{s}(c)$ that satisfies the properties required by Definition~\ref{def:fairly}. Our
definition depends on a parameter $\alpha > 0$ which we fix later. To define
$\hat{s}$, we partition the set $C$ of all possible commitments into {\em disjoint}
sets $C = R \cup\bigcup_{s,i} C_{s,i}$ that satisfy the following three properties for every $i$ and every~$s$: 
$$
C_{s,i} \subseteq f^{-1}_i(\{s\}) \, , \;\; p(c\in C_{s,i})\geq\alpha \;\,\text{or}\;\, C_{s,i} = \emptyset \, , \;\; \text{and} \;\; p(c\in R\land f_i(c)=s) < \alpha \, .
$$
The second property implies that there are at most $\alpha^{-1}$ non-empty sets
$C_{s,i}$. It is easy to see that such a partitioning exists: Start with
$R=C$ and while there exist $s$ and $i$ with 
$p(c\in R\land f_i(c) = s)\geq\alpha$, let
$C_{s,i} = \{ c\in R\mid f_i(c)=s \}$ and remove the elements of $C_{s,i}$
from $R$. 
For any $c \in C$, we now define $\hat{s}(c)$ as follows. 
We set $\hat{s}(c) = s$ for $c\in C_{s,i}$ and $\hat{s}(c) = 0$ for
$c \in R$.

Now fix some opening strategy $\dopen_{PQ}^i$ and a string $s_\circ$, 
and write $s_i$ for the verifier's output. Using
$C_{\neq s_\circ}$ as a shorthand for
$\bigcup_{s\neq s_\circ} \bigcup_{j} C_{s,j}$, we note that if $\hat{s}(c)\neq s_\circ$ then $c \in R \cup C_{\neq s_\circ}$. Thus, it follows that 
\begin{align*}
  p(s_i\neq\hat{s}(c)\land s_i = s_\circ)
  &= p(\hat{s}(c)\neq s_\circ \land s_i = s_\circ)\\
  &\leq p\big(c\in (R \cup C_{\neq s_\circ})
    \land f_i(c) = s_\circ\big)\\
  &= p(c\in R\land f_i(c) = s_\circ)
    + \sum_{s\neq s_\circ, j} p(c \in C_{s,j} \land f_i(c) = s_\circ)\\
  &\leq p(c\in R\land f_i(c) = s_\circ)
    + \!\!\!\sum_{\substack{s\neq s_\circ, j\\\text{ s.t. } C_{s,j}\neq\emptyset}}\!\!\!
    p(f_j(c) = s \land f_i(c) = s_\circ)\\[-1ex]
  &< \alpha + \alpha^{-1}\cdot\eps
\end{align*}
where the final inequality holds because $p(c\in R\land f_i(c)=s_\circ) < \alpha$ by the choice of $R$, because $p(f_j(c) = s \land f_i(c) = s_\circ) \leq \eps$ 
by the assumed binding property, and because the number of non-empty $C_{s,j}$s 
is at most $1/\alpha$. It is
easy to see that the upper bound $\alpha + \alpha^{-1}\cdot\eps$ is minimized
by setting $\alpha = \sqrt{\eps}$. We conclude that
$p(s_i\neq\hat{s}(c) \land s_i = s_\circ) < 2\sqrt{\eps}$.
\qed
\end{proof}
Combining Proposition~\ref{prop:csst_sim_open} and Theorem \ref{thm:sim_open2fairly},
we obtain the following statement for the (fairly-)binding property of the 2-prover string commitment scheme $\CHSH^n$.
\begin{corollary}\label{cor:CSST}
$\CHSH^n$ is $2^{-\frac{n}{2} + 1}$-fairly-binding.
\end{corollary}
For the fairly-weak-binding property, we can get a slightly better parameter.
Note that we do not require $\open_V$ to be deterministic here. The
proof of the theorem below is given in Appendix~\ref{ap:sim_open2fairlyweak}.
\begin{theorem}
  \label{thm:sim_open2fairlyweak}
Every 2-prover commitment scheme $\mathcal S$ that is $\eps$-fairly-binding in the sense of simultaneous opening (against classical attacks) is
$\sqrt{2\eps}$-fairly-weak-binding (against classical attacks). 
\end{theorem}

\begin{corollary}\label{cor:CSST_weak}
$\CHSH^n$ is $2^{-\frac{n-1}{2}}$-fairly-weak-binding.
\end{corollary}

\begin{remark}
It is not too hard to see that Corollary~\ref{cor:CSST_weak} above implies an upper bound on the classical value $\omega$ of the game ${\sf CHSH}_{2^n}$ considered in~\cite{BS15} of $\omega({\sf CHSH}_{2^n}) \leq 2^{-\frac{n-1}{2}} + 2^{-n}$. As such, Theorem 1.3 in~\cite{BS15} implies that the above $\eps$ is asymptotically optimal for odd $n$, i.e., the square root loss to the binding property of the {\em bit} commitment version is unavoidable (for odd $n$). 

As for security against quantum attacks, we point out that~\cite{BS15,RAM15} provide an upper bound on the quantum value $\omega^*({\sf CHSH}_q)$ of general finite-field CHSH; however, this does not directly imply security against quantum attacks of $\CHSH^n$ as a (fairly-weak-binding) string commitment scheme. 
\end{remark}

\section{Composing Commitment Schemes}
\label{sec:comp}

\subsection{The Composition Operation}

We consider two 2-prover commitment schemes $\cal S$ and ${\cal S}'$ of a restricted form, and we compose them to a new 2-prover commitment scheme ${\cal S}'' = {\cal S} \star {\cal S}'$ in a well-defined way; our composition theorem then shows that ${\cal S}''$ is secure (against classical attacks) if $\cal S$ and ${\cal S}'$ are. We start by specifying the restriction to $\cal S$ and ${\cal S}'$ that we impose. 

\begin{definition}
Let ${\cal S}$ and ${\cal S}'$ be two 2-prover string commitment schemes. We call the pair $(\cal S,{\cal S}')$ {\em eligible} if the following three properties hold, or they hold with the roles of $P$ and $Q$ exchanged.
\begin{enumerate}\setlength{\parskip}{1ex}
\item The commit phase of $\cal S$ is a protocol $\com_{PV} = (\com_P,\com_V)$ between $P$ and $V$ only, and the opening phase of $\cal S$ is a protocol $\open_{QV} = (\open_Q,\open_V)$ between $Q$ and $V$ only. In other words, $\com_Q$ and $\open_P$ are both trivial and do nothing.%
\footnote{Except that $\com_Q$ may
output state information to the opening protocol $\open_Q$, e.g., in order
to pass on the commit phase randomness.}
Similarly, the commit phase of ${\cal S}'$ is a protocol $\com'_{QV}$ between $Q$ and $V$ only (but both provers may be active in the opening phase).

\item The opening phase $\open_{QV}$ of $\cal S$ is of the following simple form: $Q$ sends a bit string $y \in \set{0,1}^m$ to $V$, and $V$ computes $s$ deterministically as $s = \Extr(y,c)$, where $c$ is the commitment.%
\footnote{Our composition theorem also works for a randomized $\Extr$, but
for simplicity, we restrict to the deterministic case.}
\item The domain of ${\cal S}'$ contains (or equals) $\set{0,1}^m$. 
\end{enumerate}
Furthermore, we specify that the allowed attacks on $\cal S$ are so that
$P$ and $Q$ do not communicate during the course of the entire scheme, and the allowed attacks on $\cal S'$ are so that $P$ and $Q$ do not communicate during the course of the commit phase but there may be limited communication during the opening phase.
\end{definition}

\noindent
An example of an eligible pair of 2-prover commitments is the pair $(\CHSH^n,\XCHSH^n)$, where $\XCHSH^n$ coincides with scheme $\CHSH^n$ except that the roles of $P$ and $Q$ are exchanged. 

\begin{remark}
For an eligible pair $(\cal S,{\cal S}')$, it will be convenient to understand $\open_Q$ and $\open_V$ as {\em non-interactive} algorithms, where $\open_Q$ produces $y$ as its {\em output}, and $\open_V$ takes $y$ as additional {\em input} (rather than viewing the pair as a protocol with a single one-way communication round). 
\end{remark}
We now define the composition operation. Informally, committing is done by means of committing using $\cal S$, and to open the commitment, $Q$ uses $\open_Q$ to locally compute the opening information $y$ and he commits to $y$ with respect to the scheme ${\cal S}'$, and then this commitment is opened (to $y$), and $V$ computes and outputs $s = \Extr(y, c)$. Formally, this is captured as follows (see also Figure~\ref{fig:comp}). 

\begin{definition}
Let ${\cal S} = (\com_{PV},\open_{QV})$ and ${\cal S}' = (\com'_{QV},\open'_{PQV})$ be an eligible pair of 2-prover commitment schemes. Then, their composition ${\cal S} \star {\cal S}'$ is defined as the 2-prover commitment scheme consisting of $\com_{PV} = (\com_P[\rndc_{PQ}],\com_V)$ and 
$$
\open''_{PQV} = (\open'_P,\, \open'_Q \circ \com'_{Q} \circ \open_Q,\, \open_V \circ \open'_V \circ \com'_V) \, ,
$$
where we make it explicit that $\com_P$ and $\open_Q$ use joint randomness, and so do $\com'_{Q}$ and $\open'_P$. 

When considering attacks against the binding property of the composed scheme ${\cal S} \star {\cal S}'$, we declare that the allowed deterministic attacks%
\footnote{The allowed \em randomized \em attacks are then naturally given as those that pick one of the deterministic attacks according to some distribution.}
 are those of the form $(\dcom_P,\dopen'_{PQ}\circ\ptoq_{PQ}\circ\dcom'_Q)$, where $\dcom_P$ is an allowed deterministic commit strategy for $\cal S$, $\dcom'_Q$ and $\dopen'_{PQ}$ are allowed deterministic commit and opening strategies for ${\cal S}'$, and $\ptoq_{PQ}$ is the one-way communication protocol that communicates $P$'s input to~$Q$ (see also Figure~\ref{fig:dishonest}).%
 \footnote{This one-way communication models that in the relativistic setting, sufficient time has passed at this point for $P$ to inform $Q$ about what happened during $\com_P$. }
\end{definition}

\begin{figure}
\begin{center}
\ifpdf
\includegraphics[scale=0.45]{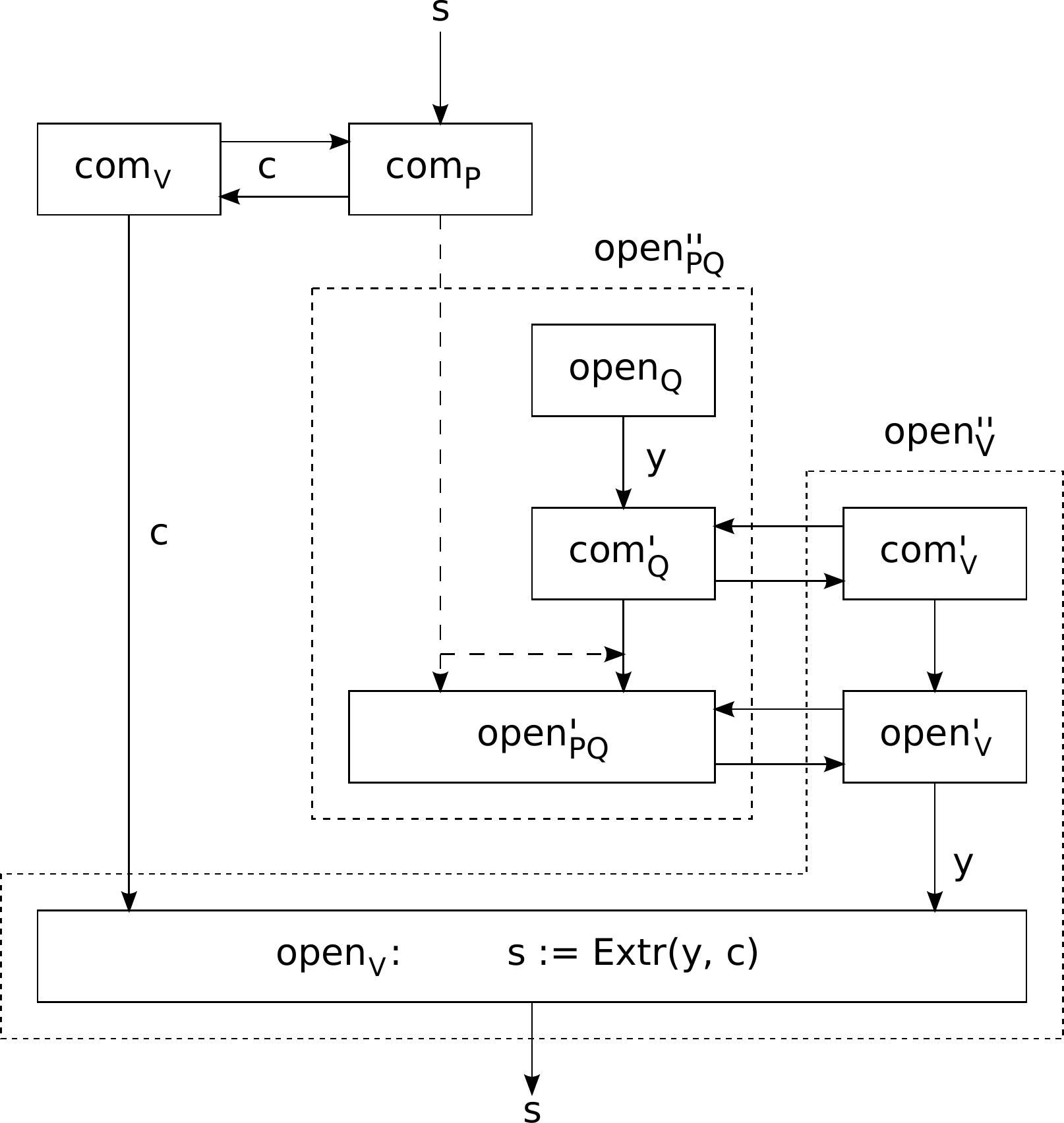} 
\else
\includegraphics[scale=0.45]{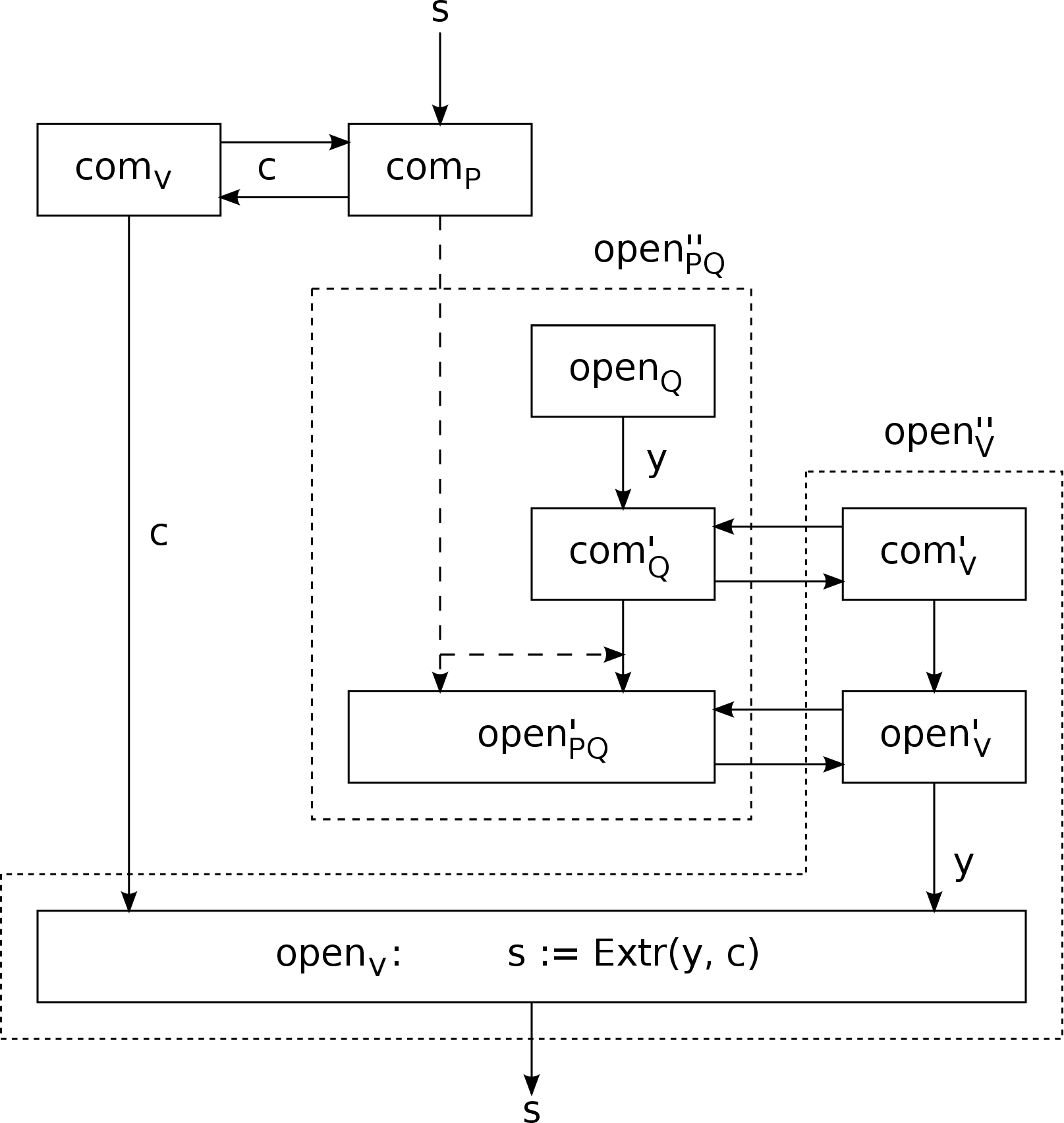} 
\fi
\end{center}
\caption{The composition of ${\cal S}$ and ${\cal S}'$ (assuming single-round commit phases). The dotted arrows indicate communication allowed to the dishonest provers. }\label{fig:comp}
\end{figure}

\begin{remark}
It is immediate that ${\cal S} \star {\cal S}'$ is a commitment scheme in the
sense of Definition~\ref{def:com}, and that it is complete if ${\cal S}$ and ${\cal S}'$ are, with the error parameters adding up. Also, the hiding property is obviously inherited from $\cal S$; however, the point of the composition is to keep the hiding property alive for longer, namely up to before the last round of the opening phase\,---\,recall that, using the terminology used in context of relativistic commitments, these rounds of the opening phase up to before the last would then be referred to as the {\em sustain phase}. We show in Appendix~\ref{ap:hiding} that ${\cal S} \star {\cal S}'$ is hiding up to before the last round, with the error parameters adding up. 

It is intuitively clear that ${\cal S} \star {\cal S}'$ {\em should be} binding if ${\cal S}$ and ${\cal S}'$ are: Committing to the opening information $y$ and then opening the commitment allows the provers to {\em delay} the announcement of $y$ (which is the whole point of the exercise), but it does not allow them to {\em change} $y$, by the binding property of ${\cal S}'$; thus, ${\cal S} \star {\cal S}'$ should be (almost) as binding as $\cal S$. This intuition is confirmed by our composition theorem below. 
\end{remark}

\begin{remark}
We point out that the composition ${\cal S} \star {\cal S}'$ can be naturally defined for a {\em larger} class of pairs of schemes (e.g.~where {\em both} provers are active in the commit phase of both schemes), and the above intuition still holds. However, our proof only works for this restricted class of (pairs of) schemes. Extending the composition result in that direction is an open problem. 
\end{remark}

\begin{remark}
\label{rem:selfcomp}
We observe that if $({\cal S},{\cal XS})$ is an eligible pair, where $\cal XS$ coincides with $\cal S$ except that the roles of $P$ and $Q$ are exchanged, then so is $({\cal XS},{\cal S} \star {\cal XS})$. 
As such, we can then compose $\cal XS$ with ${\cal S} \star {\cal XS}$, and obtain yet another eligible pair $({\cal S},{\cal XS} \star {\cal S} \star {\cal XS})$, etc.
We write $\cal S_m$ for the $m$-fold composition of $\cal S$ with itself, i.e.,
$\cal S_m = {\cal S} \star {\cal XS} \star {\cal S} \star \ldots$ for $m$ terms.
Applying this to the schemes ${\cal S} = \CHSH^n$, we obtain the multi-round scheme from Lunghi \etal.~\cite{LKB+15}. As such, our composition theorem below implies security of their scheme\,---\,with a {\em linear} blow-up of the error term (instead of double exponential). 

We point out that formally we obtain security of the Lunghi \etal. scheme as a {\em 2-prover commitment scheme} under an {\em abstract restriction} on the provers' communication: In every round, the active prover cannot access the message that the other prover received in the previous round. As such, when the rounds of the protocol are executed fast enough so that it is ensured that there is no time for the provers to communicate between subsequent rounds, then security as a {\em relativistic commitment scheme} follows immediately. 
\end{remark}
Before stating and proving the composition theorem, we need to single out one more relevant parameter. 

\begin{definition}
Let $({\cal S},{\cal S}')$ be an eligible pair, which in particular means that $V$'s action in the opening phase of $\cal S$ is determined by a function $\Extr$. 
We define $k({\cal S}) := \max_{c,s}|\Set{y}{\Extr(y,c) = s}|$. 
\end{definition}
I.e., $k({\cal S})$ counts the number of $y$'s that are consistent with a given string $s$ (in the worst case). 
Note that $k(\CHSH^n) = 1$: For every $a,x,s \in \bin^n$ there is at most one $y \in \bin^n$ such that $x + y = a \cdot s$.

\subsection{The Composition Theorems}

In the following composition theorems, we take it as understood that the assumed respective binding properties of $\cal S$ and ${\cal S}'$ hold with respect to a well-defined respective classes of allowed attacks. 
We start with the composition theorem for the fairly-binding property, which
is easier to prove than the one for the fairly-weak-binding property.

\begin{theorem}\label{thm:comp}
Let $({\cal S},{\cal S}')$ be an eligible pair of 2-prover commitment schemes, and assume that $\cal S$ and ${\cal S}'$ are respectively $\eps$-fairly-binding
and $\delta$-fairly-binding. Then, their composition ${\cal S}'' = {\cal S} \star {\cal S}'$ is $(\eps+k({\cal S})\cdot\delta)$-fairly-binding. 
\end{theorem}

\begin{proof}
We first consider the case $k({\cal S})=1$. 
We fix an attack $(\dcom_P,\dopen''_{PQ})$ against ${\cal S}''$. Without loss of generality, the attack is deterministic, so $\dopen''_{PQ}$ is of the form $\dopen''_{PQ} = \dopen'_{PQ}\circ\ptoq_{PQ}\circ\dcom'_Q$.

Note that $\dcom_P$ is also a commit strategy for $\cal S$. As such, by the fairly-binding property of $\cal S$, there exists a function $\hat{s}(c)$, only depending on $\dcom_P$, so that the property specified in Definition~\ref{def:fairly} is satisfied for every opening strategy $\dopen_Q$ for $\cal S$. We will show that it is also satisfied for the (arbitrary) opening strategy \smash{$\dopen''_{PQ}$} for ${\cal S}''$, except for a small increase in $\eps$: We will show that \smash{$p(\hat{s}(c) \neq s \wedge s = s_\circ) \leq \eps + \delta$} for every fixed target string $s_\circ$. This then proves the claim. 

To show this property on $\hat{s}(c)$, we ``decompose and reassemble'' the
attack strategy $(\dcom_P,\dopen'_{PQ}\circ\ptoq_{PQ}\circ\dcom'_Q)$ for
${\cal S}''$ into an attack strategy $(\dcom'_Q,\dnewopen'_{PQ})$ for
${\cal S}'$ with $\dnewopen'_{PQ}$ formally defined as
$$
\dnewopen'_{PQ}[c](\overline{state}'_Q) := \dopen'_{PQ}\bigl(\overline{state}_P(c)\|(\overline{state}_P(c),\overline{state}'_Q)\bigr)
$$
where 
$$
(\overline{state}_P(c)\|c) \leftarrow \bigl(\dcom_P||\com_V\bigr) \, .
$$
Informally, this means that ahead of time, $P$ and $Q$ {\em simulate} an execution of $(\dcom_P(\emptyset)||\com_V(\emptyset))$ and take the resulting communication/commitment%
\footnote{Recall that by convention (Remark~\ref{rem:entire_com}), the commitment $c$ equals the communication between $V$ and, here,~$P$. } 
$c$ as shared randomness, and then $\dnewopen'_{PQ}$ computes $\overline{state}_P$ from $c$ as does $\dcom_P$, and runs $\dopen'_{PQ}$ (see Figure~\ref{fig:dishonest}).\footnote{We are using here that $Q$ is
inactive during $\dcom_{PQ}$ and $P$ during $\dcom'_{PQ}$, and thus the two
``commute''. } 
It follows from the fairly-binding property that there is a function $\hat{y}(c')$
of the commitment $c'$ so that
$p(\hat{y}(c') \neq y \,\wedge\, y = y_\circ(c)) \leq \delta$ for every
function $y_\circ(c)$. 

\begin{figure}
\begin{center}
\ifpdf
\includegraphics[scale=0.45]{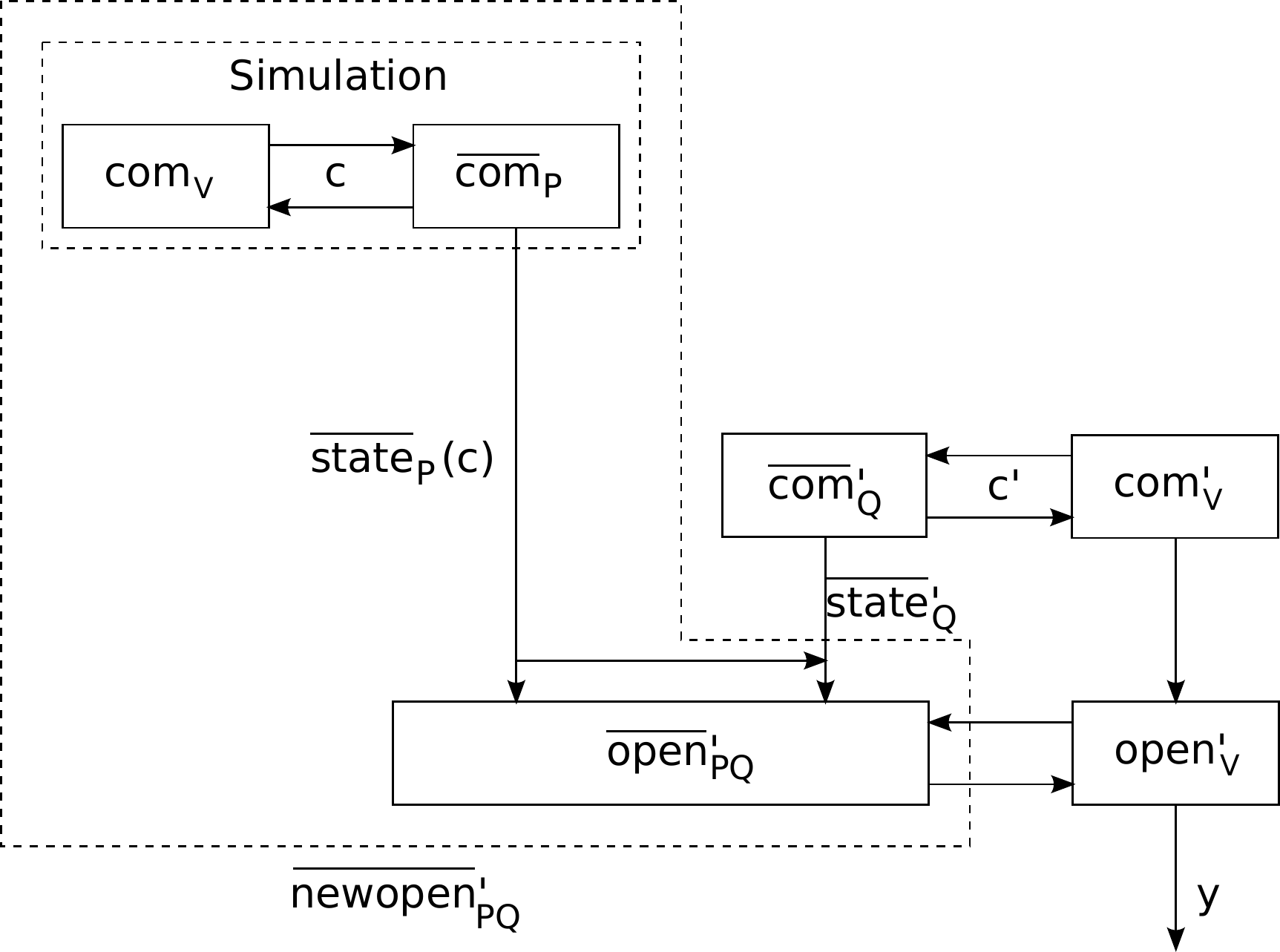} 
\else
\includegraphics[scale=0.45]{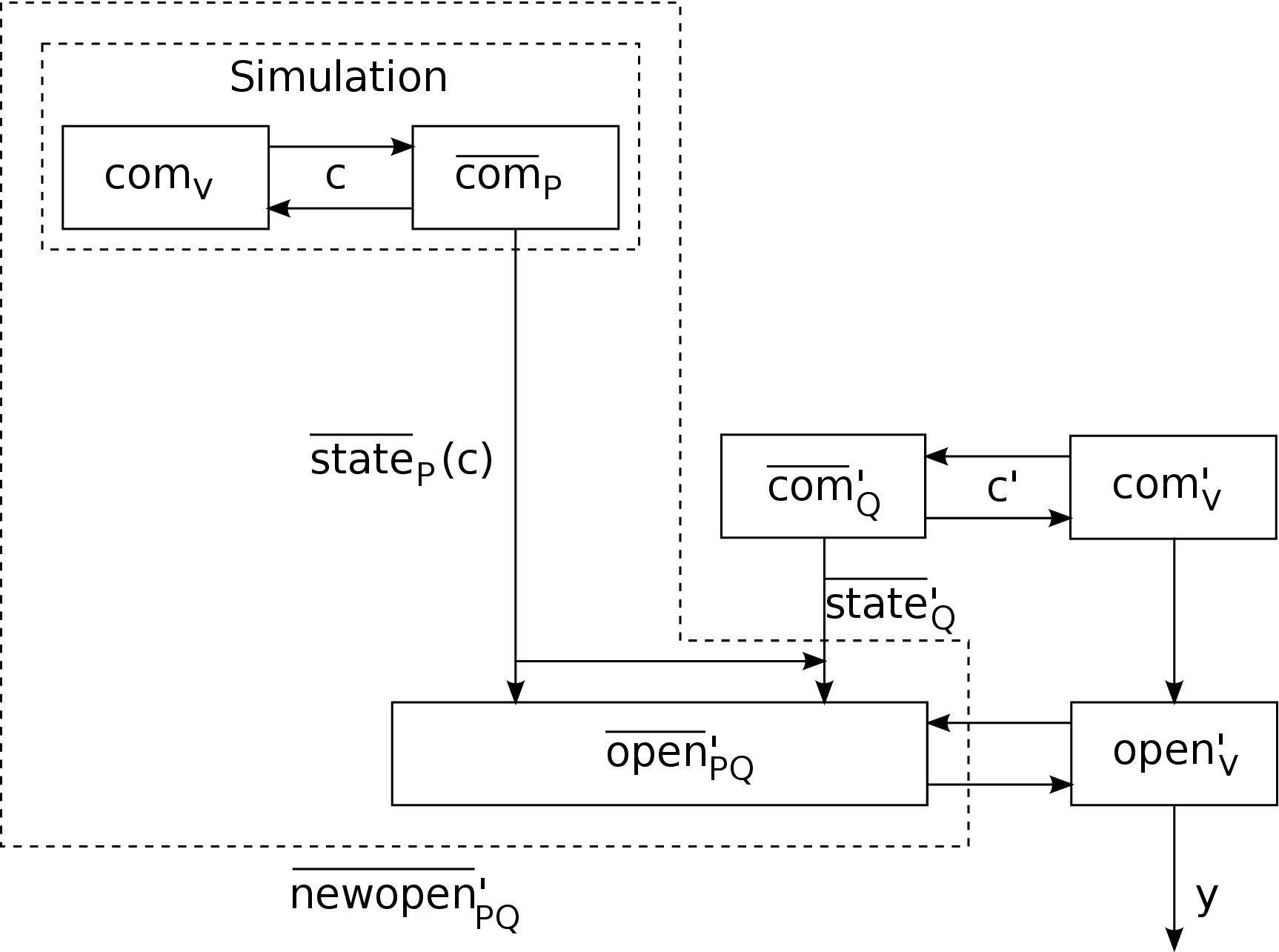} 
\fi
\end{center}
\caption{Constructing the opening strategy $\dnewopen'_{PQ}$ against ${\cal S}'$. }\label{fig:dishonest}
\end{figure}

The existence of $\hat{y}$ now gives rise to an opening strategy
$\dopen_Q$ for $\cal S$; namely, simulate the commit phase of $\mathcal S'$ to
obtain the commitment $c'$, and output $\hat{y}( c')$.
By Definition~\ref{def:fairly}, for $\tilde{s} := \Extr(\hat{y}(c'),c)$ and every
$s_\circ$, $p(\hat{s}(c) \neq \tilde{s} \wedge \tilde{s} = s_\circ) \leq \eps$.

We are now ready to put things together. 
Fix an arbitrary target string $s_\circ$. For any $c$ we let
$y_\circ(c)$ be the unique string such that
$\Extr(y_\circ(c),c) = s_\circ$ (and some default string if no
such string exists); recall, we assume for the moment that $k({\cal S})=1$. 
Omitting the arguments in $\hat{s}(c),\hat{y}(c')$ and $y_\circ(c)$, 
it follows that
\begin{align*}
p(\hat{s}& \neq s  \wedge s = s_\circ) 
\,\leq\, p(\hat{s} \neq s \wedge s = s_\circ \wedge s = \tilde{s}) + p(s=s_\circ \wedge s \neq \tilde{s}) \\
&\leq\, p(\hat{s} \neq \tilde{s} \wedge \tilde{s} = s_\circ) + p\bigr(\Extr(y,c) \neq  \Extr(\hat{y},c)
\wedge\, \Extr(y,c) = s_\circ\bigl) \\
&\leq\, p(\hat{s} \neq \tilde{s} \wedge \tilde{s} = s_\circ) + p(y \neq \hat{y} \wedge y = y_\circ) \\
&\leq\, \eps + \delta\text .
\end{align*}
Thus, $\hat{s}$ is as required. 

For the general case where $k({\cal S}) > 1$, we can reason similarly, except that we then list the $k \leq k({\cal S})$ possibilities $y^1_\circ(c),\ldots,y^{k}_\circ(c)$ for $y_\circ(c)$, and conclude that 
$
  p(s\neq\tilde{s} \land s = s_\circ)
  \leq\sum_{i} p\bigl(y\neq\hat{y}\land y=y^i_\circ\bigr)
  \leq k(\mathcal S) \cdot \delta
$,
which then results in the claimed bound. 
\qed
\end{proof}

\begin{remark}
\label{rem:fairly_to_strong}
Putting things together, we can now conclude the security (i.e., the binding property) of the Lunghi \etal.\ multi-round commitment scheme. Corollary~\ref{cor:CSST} ensures the fairly-binding property of $\CHSH^n$, i.e., the Cr{\'e}peau \etal. scheme as a string commitment scheme, with parameter $2^{-n/2 + 1}$. The composition theorem (Theorem~\ref{thm:comp}) then guarantees the fairly-binding property of the $m$-fold composition as a string commitment scheme, with parameter \mbox{$(m+1)\cdot 2^{-n/2 + 1}$}. Finally, Proposition~\ref{prop:fairlytostrong} implies
that the $m$-fold composition of $\CHSH^n$ with itself is a $\eps_m$-binding bit commitment scheme with error parameter $\eps_m = (m+1)\cdot 2^{-n/2+2}$ as claimed in the introduction, or, more generally, and by taking Remark~\ref{rem:genfairlytostrong} into account, a $(m+1)\cdot 2^{-n/2+k+1}$-binding $k$-bit-string commitment scheme. 
\end{remark}

For completeness, we also show the composition theorem for the weak version of the binding property. Since this notion makes sense also against quantum attacks, we emphasize the restriction to classical attacks\,---\,extending the theorem to quantum attacks is an open problem.  

\begin{theorem}\label{thm:comp_weak}
Let $({\cal S},{\cal S}')$ be an eligible pair of 2-prover commitment schemes, and assume that $\cal S$ and ${\cal S}'$ are respectively $\eps$-fairly-weak-binding
and $\delta$-fairly-weak-binding against classical attacks. Then, their composition ${\cal S}'' = {\cal S} \star {\cal S}'$ is a $(\eps+k({\cal S})\cdot\delta)$-fairly-weak-binding 2-prover commitment scheme against classical attacks. 
\end{theorem}

\begin{proof}
We first consider the case $k({\cal S})=1$. 
We fix an arbitrary deterministic attack $(\dcom_P,\dopen''_{PQ})$ against ${\cal S}''$, where $\dopen''_{PQ}$ is of the form $\dopen''_{PQ} = \dopen'_{PQ}\circ\ptoq_{PQ}\circ\dcom'_Q$. 
Let $a$ be $V$'s randomness in $\com_V$. Then, $c$ is a function $c(a)$ of $a$, and the distribution $p(a,y)$ is well defined.  
Since $\dcom_P$ is also
an attack strategy against $\cal S$, there exists a distribution $p(\hat{s})$
(only depending on $\dcom_P$) such that Definition~\ref{def:fairly_weak} is
satisfied for every opening strategy $\dopen_Q$ for $\cal S$.

Similar to the proof of Theorem~\ref{thm:comp}, we reassemble the attack strategy $(\dcom_P,\dopen'_{PQ}\circ\ptoq_{PQ}\circ\dcom'_Q)$ for ${\cal S}''$ into an attack strategy $(\dcom'_Q,\dnewopen'_{PQ})$ for ${\cal S}'$. 
Concretely, for every fixed choice of $a$,
we obtain a deterministic opening strategy $\dnewopen'_{PQ,a}$ given by
$$
\dnewopen'_{PQ,a}(\overline{state}'_Q) := \dopen'_{PQ}\bigl(\overline{state}_P(c(a))\|(\overline{state}_P(c(a)),\overline{state}'_Q)\bigr) \, ,
$$
and the distribution of the verifier's output $y$ when the provers
use $\dnewopen'_{PQ,a}$ is $p(y|a)$.
It follows from the fairly-weak-binding property of ${\cal S}'$ that there exists a distribution $p(\hat{y})$, only depending on $\dcom'_Q$, so that for every choice of $a$ there exists a consistent joint distribution $p(\hat{y},y|a)$ so that $p(\hat{y} \neq y \,\wedge\, y = y_\circ|a) \leq \delta$ for every fixed target string $y_\circ$. 
Note that here, consistency in particular means that $p(\hat{y}|a) = p(\hat{y})$.
This joint conditional distribution $p(\hat{y},y|a)$ 
together with the distribution $p(a)$ of $a$ then naturally defines the distribution $p(a,\hat{y},y)$, which is consistent with $p(a,y)$ considered above.

The existence of $p(\hat{y})$ now gives rise to an opening strategy $\dopen_Q$ for $\cal S$; namely, sample $\hat{y}$ according to $p(\hat{y})$ and output $\hat{y}$.
Note that the joint distribution of $a$ and $\hat{y}$ in this ``experiment'' is given by
$$
p(a)  \cdot  p(\hat{y}) = p(a) \cdot p(\hat{y}|a) = p(a,\hat{y}) \, ,
$$
i.e., is consistent with the distribution $p(a,\hat{y},y)$ above.
By Definition~\ref{def:fairly_weak}, we know there exists a joint distribution $p(\hat{s},\tilde{s})$, consistent with $p(\hat{s})$ fixed above and with $p(\tilde{s})$ determined by $\tilde{s} := \Extr(\hat{y},c(a))$, and such that $p(\hat{s} \neq \tilde{s} \wedge \tilde{s} = s_\circ) \leq \eps$ for every $s_\circ$. 
We can now ``glue together'' $p(\hat{s},\tilde{s})$ and $p(c,\hat{y},y,\tilde{s})$, i.e., find a joint distribution that is consistent with both, by setting
$$
p(a,\hat{y},y,\tilde{s},\hat{s}) := p(a,\hat{y},y,\tilde{s}) \cdot p(\hat{s}|\tilde{s}) \, .
$$
We now fix an arbitrary target string $s_\circ$. Furthermore, for any $a$ we let $y_\circ(a)$ be the unique string such that $\Extr(y_\circ(a),c(a)) = s_\circ$ (and to some default string if no such string exists); recall, we assume for the moment that $k({\cal S})=1$. 
With respect to the above joint distribution, it then holds that
\begin{align*}
p(\hat{s} \neq s  \wedge s = s_\circ) 
\,&=\, p(\hat{s} \neq s \wedge s = s_\circ \wedge s = \tilde{s}) + p(s=s_\circ \wedge s = s_\circ \wedge s \neq \tilde{s}) \\
&\leq\, p(\hat{s} \neq \tilde{s} \wedge s = s_\circ \wedge s = \tilde{s}) + p(s \neq \tilde{s} \wedge s = s_\circ) \\
&\leq\, p(\hat{s} \neq \tilde{s} \wedge \tilde{s} = s_\circ) + p\bigr(\Extr(y,c(a)) \neq  \Extr(\hat{y},c(a)) \,\wedge\, \Extr(y,c(a)) = s_\circ\bigl) \\
&\leq\, p(\hat{s} \neq \tilde{s} \wedge \tilde{s} = s_\circ) + p(y \neq \hat{y} \wedge y = y_\circ (a)) \\
&\leq\, p(\hat{s} \neq \tilde{s} \wedge \tilde{s} = s_\circ) + \textstyle\sum_a p(a) \cdot p(y \neq \hat{y} \wedge y = y_\circ(a)|a) \\
&\leq\, \eps + \delta \, .
\end{align*} 
Thus, the distribution $p(\hat{s},s)$ is as required. 

For the case where $k({\cal S}) > 1$, we can reason similarly, except that we then list the $k \leq k({\cal S})$ possibilities $y^1_\circ(a),\ldots,y^{a}_\circ(a)$ for $y_\circ(a)$, and conclude that 
$
  p(s\neq\tilde{s} \land s = s_\circ)
  \leq\sum_{i} p\bigl(y\neq\hat{y}\land y=y^i_\circ(a)\bigr)
  \leq k(\mathcal S) \cdot \delta
$,
which then results in the claimed bound. 
\qed
\end{proof}

\begin{remark}
\label{rem:fairlyweak_to_strong}
Analogously to Remark~\ref{rem:fairly_to_strong}, we can conclude from
Corollary~\ref{cor:CSST_weak} and Theorem~\ref{thm:comp_weak} that $\CHSH^n$ is
$(m+1)\cdot 2^{-(n-1)/2}$-fairly-weak-binding. It follows from
Proposition~\ref{prop:fairlyweaktostrong} that $\CHSH^n$ is a
$(m+1)\cdot 2^{-(n+1)/2}$-weak-binding bit-commitment scheme. More generally,
we can conclude that for any $k < n$, it is a
$(m+1)\cdot 2^{-(n-1)/2 + k}$-weak-binding $k$-bit string commitment scheme.
Below, we show how to avoid the factor $2$ introduced by invoking
Proposition~\ref{prop:fairlyweaktostrong}. 
\end{remark}

\subsection{Variations}

In this section, we show two variants of the composition theorems. The first
one says that if we compose a weak-binding with a fairly-weak-binding scheme, we
obtain a weak-binding scheme. This allows us to slightly improve the parameter
in Remark~\ref{rem:fairlyweak_to_strong}. The proof crucially relies on the
fact that, in the weak definition, there is some freedom in ``gluing together''
the distributions $p(s)$ and $p(\hat{s})$.
The second variant says that composing two binding (or weak-binding) schemes
yields a binding (or weak-binding, respectively) scheme.

We start by proving the following two properties for fairly-weak-binding commitment schemes. The first property shows that one may assume the joint distribution $p(\hat{s}, s)$ to be such that $s$ and $\hat{s}$ are independent conditioned on $s\neq\hat{s}$. 

\begin{lemma}
\label{lm:indep}
Let $\mathcal S$ be a $\eps$-fairly-weak-binding commitment scheme. Then, for any $\dcom _{PQ}$ and $\dopen _{PQ}$ there exists a joint distribution $p(\hat{s}, s)$ as required by Definition~\ref{def:fairly_weak}, but with the additional property that
$$
p(\hat{s}, s|s\neq\hat{s}) = p(\hat{s}|s\neq\hat{s}) \cdot p(s|s\neq\hat{s}) \, .
$$
\end{lemma}

\begin{proof}
Since the scheme is $\eps$-fairly-weak-binding, it follows that there exists a
consistent joint distribution $p(\hat{s}, s)$ such that $p(s\neq\hat{s}\land s = s_\circ)\leq\eps$ for every $s_\circ$. Because of this, we have
$$
  p(s=s_\circ) = p(s\!=\!s_\circ\land \hat{s} \!=\! s_\circ)
    + p(s \!=\! s_\circ\land \hat{s}\!\neq\! s_\circ)
  = p(s \!=\! s_\circ\land \hat{s} \!=\! s_\circ) + p(s\!\neq\!\hat{s}\land s \!=\! s_\circ)
  \leq p(\hat{s} = s_\circ) + \eps\text .
$$
We apply Lemma~\ref{lm:eps_dist} to the marginal distributions $p(\hat{s})$
and $p(s)$. The resulting joint distribution $\tilde{p}(\hat{s}, s)$ satisfies
$
  \tilde{p}(\hat{s} = s_\circ \land s = s_\circ|s = \hat{s})
  = \min \{p(s=s_\circ),p(\hat{s} = s_\circ)\}
$ and
$
  \tilde{p}(\hat{s}, s|s\neq\hat{s})
  = \tilde{p}(\hat{s}|s\neq\hat{s}) \cdot \tilde{p}(s|s\neq\hat{s})
$. It remains to show that $\tilde{p}(s\neq\hat{s}\land s = s_\circ)\leq\eps$
for all $s_\circ$. Indeed, we have
\begin{align*}
  \tilde{p}(s\neq\hat{s}\land s = s_\circ)
  &= \tilde{p}(s=s_\circ) - \tilde{p}(s = \hat{s} \land s = s_\circ)\\
  &= \tilde{p}(s=s_\circ) - \tilde{p}(\hat{s} = s_\circ \land s = s_\circ)\\
  &= p(s=s_\circ) - \min \{p(\hat{s} =s_\circ), p(s = s_\circ)\} \\
  &\leq p(s=s_\circ) - (p(s=s_\circ) - \eps)\\
  &= \eps
\end{align*}
as claimed.
\qed
\end{proof}
The second property shows that the quantification over all {\em fixed} $s_\circ$ in Definition~\ref{def:fairly_weak} of the fairly-weak-binding property can be relaxed to $s_\circ$ that may depend on $\hat{s}$, but only on $\hat{s}$. Note that we can obviously not allow $s_\circ$ to depend (arbitrarily) on $s$, since then one could choose $s_\circ = s$. 

\begin{proposition}
\label{prop:indep2}
Let $\mathcal S$ be a $\eps$-fairly-weak-binding commitment scheme. Then
$$
\forall\, \dcom_{PQ} \; \exists \, p(\hat{s}) \; \forall \, \dopen_{PQ} \; \exists \, p(\hat{s},s) \; \forall \, p(s_\circ | \hat{s}) : p(s \neq \hat{s} \,\land\, s = s_\circ) \leq \eps  \, ,
$$
where it is understood that $p(\hat{s},s,s_\circ) := p(\hat{s},s) \cdot p(s_\circ|\hat{s})$. Thus, the joint distribution $p(\hat{s},s)$ is such that $p(s \neq \hat{s} \,\land\, s = s_\circ) \leq \eps$ holds in particular for any function $s_\circ = f(\hat{s})$ of $\hat{s}$. 
\end{proposition}

\begin{proof}
For given $\dcom_{PQ}$ and $\dopen_{PQ}$, let $p(\hat{s},s)$ be as guaranteed by the fairly-weak-binding property. By Lemma~\ref{lm:indep}, we may assume without loss of generality that 
$p(\hat{s}, s|s\neq\hat{s}) = p(\hat{s}|s\neq\hat{s}) \, p(s|s\neq\hat{s})$. Then, by
Lemma~\ref{lm:chains}, we also have that
$p(s,s_\circ|s\neq\hat{s}) = p(s|s\neq\hat{s})\, p(s_\circ|s\neq\hat{s})$.
It follows that
\begin{align*}
  p(s\neq\hat{s}\land s = s_\circ)
  &= p(s\neq\hat{s}) \cdot p(s = s_\circ|s\neq\hat{s}) \\[2ex]
  &= p(s\neq\hat{s}) \sum_{s_\circ^*} p(s = s_\circ^* \land s_\circ = s_\circ^*|s\neq\hat{s}) \\
  &= p(s\neq\hat{s}) \sum_{s_\circ^*} p(s = s_\circ^*|s\neq\hat{s}) \cdot p(s_\circ = s_\circ^*|s\neq\hat{s}) \\
  &= \sum_{s_\circ^*} p(s\neq\hat{s} \land s = s_\circ^*) \cdot p(s_\circ = s_\circ^* | s\neq\hat{s}) \\
  &\leq \eps \cdot \sum_{s_\circ^*} p(s_\circ = s_\circ^* | s\neq\hat{s}) \\
  &= \eps
\end{align*}
where the inequality follows from the fact that $p(s\neq\hat{s} \land s = s_\circ^*) \leq \eps$ for every {\em fixed} $s_\circ^*$.
\qed
\end{proof}
For the rest of the section, we take it as understood that we only consider
classical attacks. 

\begin{theorem}
\label{thm:comp_bit}
Let $(\mathcal S, \mathcal S')$ be an eligible pair of 2-prover commitment schemes, where 
$\mathcal S$ is $\eps$-weak-binding and $\mathcal S'$ is $\delta$-fairly-weak-binding, and let $\bin^m$ be the domain of $\cal S$. Then, the composition
$S\star\mathcal S'$ is a $(\eps+(2^m \!-\!1)\cdot k(\mathcal S) \cdot \delta)$-weak-binding 
commitment scheme. \\
In particular, if $\mathcal S$ is a {\em bit} commitment scheme then $S\star\mathcal S'$ is a $(\eps+ k(\mathcal S) \cdot \delta)$-weak-binding. 
\end{theorem}

\begin{proof}
We follow the proof of Theorem~\ref{thm:comp_weak}, up to when it comes to choosing $y_\circ$. Let us first consider the case $m = 1$, i.e., $\cal S$ is a {\em bit} commitment scheme. In that case, and assuming for the moment that $k(\mathcal S) = 1$, we let $y_\circ$ be the unique string that satisfies $\Extr(y_\circ,c) = s_\circ$, but where now $s_\circ := 1-\tilde{s}$. We emphasize that for a fixed $c$, this choice of $y_\circ$ is {\em not} fixed anymore (in contrast to the choice in the proof of Theorem~\ref{thm:comp_weak}); namely, it is a function of $\tilde{s} = \Extr(\hat{y},c)$, which in turn is a function of $\hat{y}$.
Therefore, by Proposition~\ref{prop:indep2}, it still holds that
$p(y\neq\hat{y}\land y = y_\circ|a)\leq\delta$, and we can conclude that
\begin{align*}
  p(\hat{s}\neq s\land s \neq\bot)
  &\leq p(\hat{s}\neq s\land s \neq\bot \land s = \tilde{s})
    + p(s\neq\tilde{s}\land s\neq \bot)\\
  &= p(\hat{s}\neq \tilde{s}\land s \neq\bot \land s = \tilde{s})
    + p(s\neq\tilde{s}\land s = 1-\tilde{s})\\
  &\leq p(\hat{s}\neq \tilde{s}\land \tilde{s}\neq\bot)
    + p(y\neq\hat{y}\land y = y_\circ)\\
  &\leq p(\hat{s}\neq \tilde{s}\land \tilde{s}\neq\bot)
    + \textstyle\sum_a p(a) \, p(y\neq\hat{y}\land y = y_\circ|a)\\
  &\leq \eps +  \textstyle\sum_a p(a)\,\delta\\
  &= \eps + \delta \, .
\end{align*} 

In the case that $k(\mathcal S) > 1$, we instead randomly select one of the
at most $k(\mathcal S)$ strings $y_\circ$ that satisfy
$\Extr(y_\circ, c) = s_\circ = 1-\tilde{s}$. Then, conditioned on $a$, $y_\circ$ is still independent of $y$ given $\hat{y}$, so that 
Proposition~\ref{prop:indep2} still applies, and we can argue as above, except that we get a factor $k(\mathcal S)$ blow-up from $p(s\neq\tilde{s}\land s = 1-\tilde{s}) \leq k(\mathcal S) \cdot p(y\neq\hat{y}\land y = y_\circ)$. 

Finally, for the case $m > 1$, we first pick a random $s_\circ \in \bin ^m\setminus \{ \tilde{s} \}$, and then choose $y_\circ$ such that $\Extr(y_\circ,c) = s_\circ$, uniquely or at random, depending of $k({\cal S})$. Conditioned on $a$, $y_\circ$ is still independent of $y$ given $\hat{y}$, and therefore
Proposition~\ref{prop:indep2} still applies, but now we get an additional factor $(2^m -1)$ blow-up from  
$p(s\neq\tilde{s}\land s \neq \bot) \leq (2^m -1) \, p(s\neq\tilde{s}\land s = s_\circ)$.
\qed
\end{proof}

\begin{remark}
Theorem~\ref{thm:comp_bit} allows us to slightly improve the bound we obtain in Remark~\ref{rem:fairlyweak_to_strong} on the Lunghi \etal.\ multi-round commitment scheme. 
By Theorem~\ref{thm:comp_weak}, we can compose $m$ instances of $\CHSH^n$ to obtain
a $m\cdot 2^{-(n-1)/2}\,$-fairly-weak-binding string commitment scheme. Then, we can compose the Cr{\'e}peau \etal. {\em bit} commitment scheme (i.e., the bit commitment version of $\CHSH^n$), which is $2^{-(n-1)}$-weak-binding, with this fairly-weak-binding
string commitment scheme; by Theorem~\ref{thm:comp_bit}, this composition, which is the Lunghi \etal.\ multi-round bit commitment scheme, is
$\big( m\cdot 2^{-(n-1)/2} + 2^{-(n-1)} \big)$-weak-binding.
\end{remark}
Finally, for completeness, we point out that the composition theorem also applies to two ordinary binding or weak-binding commitment schemes. 

\begin{theorem}
Let $(\mathcal S, \mathcal S')$ be an eligible pair of 2-prover commitment schemes, where 
$\mathcal S$ is $\eps$-binding and $\mathcal S'$ is $\delta$-binding. Then, the composition
$S\star\mathcal S'$ is $(\eps+\delta)$-binding. The same holds for the weak-binding property.
\end{theorem}

\begin{proof}
The proof is almost the same as in Theorem~\ref{thm:comp} or
Theorem~\ref{thm:comp_weak}, respectively, except that now there are no
$s_\circ$ and $y_\circ$, and in the end we can simply conclude that
\begin{align*}
  p(s\neq\hat{s}\land s\neq\bot)
  &\leq p(s\neq\hat{s}\land s\neq\bot\land s = \tilde{s})
    + p(s\neq\tilde{s}\land s\neq\bot)\\
  &\leq p(\tilde{s}\neq\hat{s}\land \tilde{s}\neq\bot)
    + p(y\neq\hat{y}\land y\neq\bot)\\
  &\leq \eps + \delta \, ,
\end{align*}
where the second inequality holds since $y = \bot$ implies that
$s = \Extr (y, c) = \bot$.
\qed
\end{proof}

\subsection{Tightness}

We now show that our composition result is nearly tight for $\CHSH^n$.
Let $\CHSH^n_m$ be the $m$-fold composition of $\CHSH^n$ with itself, as
defined in Remark~\ref{rem:selfcomp}. We show that for even $n$,
this composed scheme can be $\eps$-weak-binding as a bit-commitment scheme only
if $\eps\gtrsim\frac 1 4 m2^{-n/2}$. A slightly weaker result was proved in
\cite{BC16}, which shows that $\eps\gtrsim\frac 1 6 m2^{-n/2}$ for even $n$.%
\footnote{The paper states $\eps\gtrsim \frac 1 3 m 2^{-n/2}$, but their
binding definition is $p_0 + p_1 \leq 1 + \eps$; to convert their bound to our
definition (equivalent to $p_0 + p_1\leq 1 + 2\eps$), it must be multiplied by
$1/2$.}
Furthermore, we show that, as a string commitment scheme, $\CHSH_m^n$ can be
$\eps$-fairly-weak-binding only if $\eps\gtrsim \frac 1 2 m2^{-n/2}$
(for even $n$).

\begin{lemma}
\label{lm:x_and_y}
Consider functions $X_n, Y_n: \mathbb F_{2^n}\times R_n\to \mathbb F_{2^n}$.
Let
\begin{equation}
\label{eq:chsh_game}
  q_n = \max_{X_n, Y_n}\ p(X_n(a, r) + Y_n(s, r) = a\cdot s)
\end{equation}
where $a$, $s$ and $r$ are selected uniformly at random in $R_n$.
It holds that:
\begin{enumerate}
  \item There are $X_n$ and $Y_n$ such that
    $p(X_n(a, r) + Y_n(s, r) = a\cdot s) = q_n$ for \em all \em
    $a, s\in\mathbb F_{2^n}$.
  \item For even $n$, we have $q_n = \Omega\bigl( 2^{-n/2} \bigr)$.
    For odd $n$, we have $q_n = \Omega\bigl( 2^{-2n/3} \bigr)$.
\end{enumerate}
\end{lemma}
\begin{proof}
Fix $X_n'$ and $Y_n'$ that achieve the maximum in Equation~\eqref{eq:chsh_game}.
We show that there also are functions $X_n$ and $Y_n$ such that for \em any \em
$a$ and $s$, $p(X_n(a, r) + Y_n(s, r) = a\cdot s) = q_n$:
Without loss of generality, $X_n'$ and $Y_n'$ depend only on $a$ and $s$, not
on $r$. Intuitively, $X_n$ and $Y_n$ do the following: They randomize their
inputs $a$ and $s$ by adding uniformly random elements
$r_a, r_s\in\mathbb F_{2^n}$, then apply $X_n'$ and $Y_n'$, and finally remove
the random terms again from the output. Formally, we let
\begin{align*}
  X_n(a, (r_a, r_s)) &= X_n'(a+r_a) - ar_s - r_ar_s\\
  Y_n(a, (r_a, r_s)) &= Y_n'(s+r_s) - r_as
\end{align*}
For $r_a$ and $r_s$ uniformly random, we have
$p(X_n'(a+r_a) + Y_n'(s+r_s) = as + ar_s +r_ar_s + sr_a) = q_n$.
Thus, it is easy to see that
$p(X_n(a, (r_a, r_s)) + Y_n(s, (r_a, r_s)) = as) = q_n$.

The functions $X_n$ and $Y_n$ in Equation~\eqref{eq:chsh_game} describe
strategies for the $\CHSH^n$ game with classical players and $q_n$ is the
maximal winning probability that classical players can achieve in this game.
As shown in \cite{BS15}, it holds that $q_n = \Omega\bigl(2^{-n/2}\bigr)$ for
even $n$, and $q_n = \Omega\bigl(2^{-2n/3}\bigr)$ for odd $n$.
\qed
\end{proof}

The following lemma can be seen as a generalization of
Theorem~\ref{thm:equiv_weak} to string commitment schemes.
Intuitively, it bounds the winning probability of the provers in the following
game: First, they have to produce a commitment. Then, they receive a uniformly
random string $s_\circ$ and, in order to win, they have to open the commitment
to $s_\circ$. The winning probability in this game is at most $\eps + 2^{-n}$,
when the scheme is an $\eps$-fairly-weak-binding $n$-bit string commitment
scheme.

\begin{lemma}
\label{lm:prob_eq}
Let $\cal S$ be a $\eps$-fairly-weak-binding $n$-bit string commitment scheme.
Fix an allowed commit strategy $\dcom_{PQ}$ for $\cal S$ and, for each
$s_\circ\in\mathbb F_{2^n}$, an allowed opening strategy $\dopen_{PQ}(s_\circ)$.
Let $p(s|s_\circ)$ be the output distribution of $\cal S$ if the provers use
$\dcom_{PQ}$ and $\dopen_{PQ}(s_\circ)$. Let $p(s_\circ)$ be distributed
uniformly over $\mathbb F_{2^n}$. Then,
$p(s=s_\circ) := \sum_{s_\circ\in\mathbb F_{2^n}}p(s_\circ)p(s=s_\circ|s_\circ)
\leq \eps + 2^{-n}$.
\end{lemma}

\begin{proof}
Let $p(\hat s)$ be a distribution that satisfies
Equation~\eqref{eq:fairly_weak_bind} for the commit strategy $\dcom_{PQ}$.
Now consider any consistent joint distribution $p(s, \hat s|s_\circ)$. Here,
consistency also means that $p(\hat s|s_\circ) = p(\hat s)$. Thus, for
a uniformly random $s_\circ$, $p(\hat s = s_\circ) = 2^{-n}$.
By the $\eps$-fairly-weak-binding property of $\cal S$, we have
$$
  \eps \geq p(s \neq \hat s \land s = s_\circ)
    \geq p(s = s_\circ) - p(\hat s = s_\circ)
    = p(s = s_\circ) -2^{-n}
$$
and thus our claim follows.
\qed
\end{proof}

With the help of the lemma above, is easy to see that $q_n$ limits the binding
parameter of the one-round scheme $\CHSH^n$: If $P$ sends $X_n(a, r)$ and $Q$
sends $Y_n(s_\circ, r)$ for uniformly random $r$, then we have
$p(s = s_\circ|a\neq 0) = q_n$, and thus $p(s = s_\circ) \geq q_n - 2^{-n}$ for
every $s_\circ$.
Thus, by Lemma~\ref{lm:prob_eq}, $\CHSH^n$ can be $\eps$-fairly-weak-binding
only if $\eps\geq q_n - 2^{-n+1}$.
We now show that this bound scales approximately linearly with the number of
rounds.

\begin{theorem}
\label{thm:tightness}
Let $q_n$ as in Lemma~\ref{lm:x_and_y}.
For odd $m$, the $\CHSH^n_m$ commitment scheme can be
$\eps$-fairly-weak-binding as a string commitment scheme only if
$$
  \eps \geq \frac{(m+1)q_n}{2} - \frac{\bigl(m^2-1\bigr)q_n^2}{8}
    - (m+1)2^{-n}\text .
$$
If $m = o\bigl(q_n^{-1}\bigr)$, it holds that $\eps\geq\Omega(mq_n)$.
If, furthermore, $n$ is even, we have $\eps\geq\Omega\bigl(m2^{-n/2}\bigr)$;
if $n$ is odd, $\eps\geq\Omega\bigl(m2^{-2n/3}\bigr)$.
\end{theorem}
\begin{proof}
Let $X_n(a, r)$ and $Y_n(b, r)$ be functions as in Lemma~\ref{lm:x_and_y}.
We define a commit strategy $\dcom_{PQ}$ and an opening strategy
$\dopen_{PQ}(s_\circ)$ for every $s_\circ$ which aims to open to $s_\circ$.

We assume that the provers have $m$ uniformly random strings
$r_i\in\mathbb F_{2^n}$ and $(m+1)/2$ uniformly random inputs $r_i'$, $i$ odd,
for $X_n$ and $Y_n$ as shared randomness. We write $c_i = (a_i, x_i)$ for the
communication between the verifier and the active prover in round $i$, where
the $x_i$ are specified below. The dishonest provers exchange their
communications as fast as possible, so in round $i+2$, the active prover knows
$c_1,\ldots, c_i$. Let $y_0 = s_\circ$ and for $i>0$, let $y_i$ such that
$\Extr(y_i, c_i) = y_{i-1}$. Such a $y_i$ exists and is unique if $a_i \neq 0$.
We only specify our strategy for the case where the verifier's messages $a_i$
are all non-zero and assume that the provers fail to open to $s_\circ$
otherwise. One can compute $y_i$ from $c_1,\ldots,c_i$, so in round $i+2$, the
active prover can compute $y_i$.

If in any round $i$, the commitment is $(a_i, r_i+a_i\cdot y_{i-1})$, the
provers can open to $s_\circ$ simply by following the honest strategy for
$\CHSH^n_m$ from that round on. The strategy described below is such that the
provers have $(m+1)/2$ chances to bring about this situation with probability
$q_n$.
\begin{itemize}
  \item Round 1 (commit): $P$ produces a ``fake commitment''
    $x_1 = X_n(a_1, r_1')$.
  \item Round $i$, $i$ even: $Q$ computes $y_{i-1}' = Y_n(y_{i-2}, r_{i-1}')$,
    hoping that $x_{i-1} + y_{i-1}' = a_{i-1}\cdot y_{i-2}$, i.e.,
    $y_{i-1}' = y_{i-1}$. He honestly commits to $y_{i-1}'$ by computing
    $x_i = a_i\cdot y_{i-1}' + r_i$.
  \item Round $i+1$, $i$ even: $P$ checks if $y_{i-1} = y_{i-1}'$. If yes, both
    provers proceed honestly from this round on, i.e., they follow the honest
    strategy for $\CHSH^n_m$ in all subsequent rounds.\footnote{$Q$ can compute
    $y_{i-1}$ in round $i+2$ and thus he too knows whether the provers should
    proceed honestly or not.}
    If not, $P$ again produces a ``fake commitment''
    $x_{i+1} = X_n(a_{i+1}, r_{i+1}')$.
  \item Round $m+1$: $Q$ sends $y_m' = Y_n(y_{m-1}, r_m')$ to $V$.
\end{itemize}

By definition, we have $y_{i-1}' = y_{i-1}$ if and only if
$X_n(a_{i-1}, r'_{i-1}) + Y_n(y_{i-2}, r'_{i-1}) = a_{i-1}\cdot y_{i-2}$,
which happens with probability $q_n$.
In this case, we have $c_i = (a_i, r_i + a_i\cdot y_{i-1})$, so
the provers can indeed open to $s_\circ$ by proceeding honestly (ignoring
completeness errors for now).

By definition of $X_n$, $Y_n$, and $q_n$, if the provers use the strategy
$\dopen_{PQ}(s_\circ)$, then for
$$
  q = 1 - ( 1 - q_n )^{(m+1)/2}
  \geq \frac{(m+1)q_n}{2} - {(m+1)/2\choose 2}q_n^2
  = \frac{(m+1)q_n}{2} - \frac{\bigl(m^2-1\bigr)q_n^2}{8}
$$
we have $p(s = s_\circ|a_1,\ldots,a_m\neq 0) = q$.
Thus, $p(s = s_\circ) \geq q - m2^{-n}$ for all $s_\circ$. Applying
Lemma~\ref{lm:prob_eq}, we conclude that the scheme can be
$\eps$-fairly-weak-binding only if
$$
  \eps\geq q - (m+1)2^{-n}
    \geq \frac{(m+1)q_n}{2} - \frac{\bigl(m^2-1\bigr)q_n^2}{8} - (m+1)2^{-n}
$$
which is in $\Omega(mq_n)$ if $m = o\bigl(q_n^{-1}\bigr)$.
Finally, we have $\Omega(mq_n) = \Omega\bigl(m2^{-n/2}\bigr)$ if $n$ is even
and $\Omega(mq_n) = \Omega\bigl(m2^{-2n/3}\bigr)$ if $n$ is odd, by claim 2 of
Lemma~\ref{lm:x_and_y}.
\qed
\end{proof}

From the analysis in the above proof, we can also derive a version of the
theorem for the bit-commitment scheme described in
Proposition~\ref{prop:fairlyweaktostrong}.

\begin{corollary}
For even $m$, the commitment scheme $\CHSH^n_m$ can be $\eps$-binding as a
bit-commitment scheme only if
$$
  \eps \geq \frac{mq_n}{4} - \frac{(m^2-2m)q_n^2}{16} - (m+1)2^{-n}\text .
$$
If $m = o\bigl(q_n^{-1}\bigr)$, it holds that $\eps\geq\Omega(mq_n)$. If
$n$ is even, we have $\eps\geq\Omega\bigl(m2^{-n/2}\bigr)$ and if it is odd,
$\eps\geq\Omega\bigl(m2^{-2n/3}\bigr)$.
\end{corollary}

\begin{proof}
Let $\dcom_P = \com_P(0)$, i.e., $P$ produces an honest commitment to $0$.
Let $\dopen_{PQ}(0) = \open_{PQ}$, i.e., the honest opening strategy.
Since the provers play honestly, they are successful with probability
at least $1 - (m+1)2^{-n}$.

For $\dopen_{PQ}(1)$, let $s_\circ$ such that $\Extr(s_\circ, c_1) = 1$.
The provers then use the strategy in the proof of Theorem~\ref{thm:tightness}
to produce a fake commitment $c_1$ and open it to $s_\circ$. Then, we have
$$
  p(b = 1|a_1,\ldots, a_m\neq 0) \geq \frac{mq_n}{2} - \frac{(m^2-2m)q_n^2}{8}
    - 2^{-n}
$$
and thus,
$$
  p(b = 1) \geq \frac{mq_n}{2} - \frac{(m^2-2m)q_n^2}{8} - (m+1)2^{-n}\text .
$$
It follows that
$$
  p(b = 0) + p(b = 1) \geq 1 + \frac{mq_n}{2} - \frac{(m^2-2m)q_n^2}{8}
    - (m+1)2^{-n+1}
$$
and, by Theorem~\ref{thm:equiv_weak}, the scheme can be $\eps$-weak-binding
only if
$$
  \eps \geq \frac{mq_n}{4} - \frac{(m^2-2m)q_n^2}{16} - (m+1)2^{-n}\text .
$$
\qed
\end{proof}

\subsection*{Acknowledgments}

We would like to thank J\k{e}drzej Kaniewski for helpful discussions regarding \cite{LKB+15}, and for commenting on an earlier version of our work. 

\bibliographystyle{alpha}
\bibliography{FF}

\begin{appendix}

\section{Proof of Lemma~\ref{lm:eps_dist}}

We first extend the respective probability spaces given by the distributions $p(x)$ and $p(y)$ by introducing an event $\Delta$ and declaring that 
$$
p(x\!=\!x_\circ \land \Delta) =  \min \set{p(x=x_\circ), p(y=x_\circ)} = p(y\!=\!x_\circ \land \Delta)
$$
for every $x_\circ \in \cal X$. 
Note that $p(\Delta)$ is well defined (by summing over all $x_\circ$). 
As we will see below, $\Delta$ will become the event $x = y$. 
In order to find a consistent joint distribution $p(x,y)$, it suffices to find a consistent joint distribution $p(x,y|\Delta)$ for $p(x|\Delta)$ and $p(y|\Delta)$, and a consistent joint distribution $p(x,y|\neg\Delta)$ for $p(x|\neg\Delta)$ and $p(y|\neg\Delta)$. The former, we choose as 
$$
p(x=x_\circ \land y = x_\circ |\Delta) :=  \min \set{p(x=x_\circ), p(y=x_\circ)}/p(\Delta)
$$ 
for all $x_\circ \in \cal X$, 
and $p(x=x_\circ \land y = y_\circ |\Delta) := 0$ for all $x_\circ \neq y_\circ \in \cal X$, and the latter we choose as
$$
p(x=x_\circ \land y = y_\circ |\neg\Delta) := p(x=x_\circ|\neg\Delta) \cdot p(y=y_\circ|\neg\Delta)
$$
for all $x_\circ,y_\circ \in \cal X$. 
It is straightforward to verify that these are indeed {\em consistent} joint distributions, as required, so that $p(x,y) = p(x,y|\Delta) \cdot p(\Delta) + p(x,y|\neg\Delta) \cdot p(\neg\Delta)$ is also consistent. Furthermore, note that $p(x\!=\!y|\Delta) = 1$ and $p(x\!=\!y|\neg\Delta) = 0$; the latter holds because we have $p(x\!=\!x_\circ \land \Delta) =  p(x=x_\circ)$ or $p(y\!=\!x_\circ \land \Delta) =  p(y=x_\circ)$ for each $x_\circ \in \cal X$, and thus $p(x\!=\!x_\circ \land \neg\Delta) =  0$ or $p(y\!=\!x_\circ \land \neg\Delta) =  0$. As such, $\Delta$ is the event $x = y$, and therefore $p(x = y = x_\circ) = p(x\!=\!x_\circ \land \Delta) =  \min \set{p(x=x_\circ), p(y=x_\circ)} $ for every $x_\circ \in \cal X$ as required.
Finally, the claim regarding $p(x,y|x \neq y)$ holds by construction. 
\qed

\section{A Property for Conditionally Independent Random Variables}

Let $p(x,y,z)$ be a distribution, and let $\Lambda \subset {\cal X} \times {\cal Y} \times {\cal Z}$ be an event. Then, we write $x\rightarrow y\rightarrow z$ to express that $p(x,z|y) = p(x|y) \, p(z|y)$, and $x\rightarrow \Lambda \rightarrow y$ to express that $p(x,y|\Lambda) = p(x|\Lambda) \, p(y|\Lambda)$, etc. 

\begin{lemma}
\label{lm:chains}
If $x\rightarrow y\rightarrow z$
and $x\rightarrow x\neq y\rightarrow y$, then
$x\rightarrow x\neq y\rightarrow z$.
\end{lemma}
\begin{proof}
We assume that  $x\rightarrow y\rightarrow z$
and $x\rightarrow x\neq y\rightarrow y$. We first observe that
$$
  p(x,x\neq y, z|y) = p(x,x\neq y|y) \, p(z|x,y,x\neq y)
  = p(x,x\neq y|y) \, p(z|x,y)
  = p(x,x\neq y|y) \, p(z|y) \, ,
$$
which means that $(x,x\neq y)\rightarrow y \rightarrow z$, and, by summing over $x$, implies $x\neq y \rightarrow y \rightarrow z$. 
It follows that
$$
p(z|x,y,x\neq y) = p(z|y) = p(z|y,x\neq y) \, ,
$$
which actually means that $x\rightarrow (y,x\neq y)\rightarrow z$. Therefore, 
\begin{align*}
  p(x,z|x\neq y) = \sum_y p(x,y,z|x\neq y) 
  &=  \sum_y p(x,y|x\neq y) \, p(z|x,y,x\neq y)\\
  &= p(x|x\neq y) \sum_y p(y|x\neq y) \, p(z|y,x\neq y)\\
  &= p(x|x\neq y) \sum_y p(y,z|x\neq y)\\
  &= p(x|x\neq y) \, p(z|x\neq y) \, ,
\end{align*}
which was to be proven. 
\qed
\end{proof}

\section{Proof of Theorem~\ref{thm:sim_open2fairlyweak}}
\label{ap:sim_open2fairlyweak}
Fix a commit strategy $\dcom _{PQ}$ against $\mathcal S$. 
Enumerate all strings in the domain $\bin^n$ of $\cal S$ as $s_1,\ldots,s_{2^n}$, and for every $i \in \set{1,\ldots,2^n}$ let $\dopen ^i_{PQ}$ be an opening strategy maximizing $p_i := p(s = s_i)$, where $s$ is the output of the verifier when $P$ and $Q$ use this strategy. We assume without loss of generality that the $p_i$s are in descending order. 
We define $p(\hat{s})$ as follows. Let $N \geq 2$ be an integer which
we will fix later. By Definition \ref{def:sim_open} and inequality (\ref{eq:pr_sum}), it holds that
$$
  \sum _{i=1}^N p_i \leq 1+{N\choose 2} \cdot \eps = 1 + \frac{N(N-1)}{2} \cdot \eps
$$
where we let $p_i = 0$ for $i > 2^n$ in case $N > 2^n$.
We would like to define $p(\hat{s})$ as $p(\hat{s} = s_i) := p_i-(N-1)\eps/2$
for all $i\leq N, 2^n$; however, this is not always possible because $p_i-(N-1)\eps/2$ may be negative. To deal with this, let $N'$ be the largest
integer such that $N'\leq N$ and $p_1,\dots, p_{N'}\geq (N-1)\eps/2$.
(We take $N=0$ if $p_1 < (N-1)\eps/2$.) It follows that
$$
  \sum _{i=1}^{N'} p_i \leq 1 + \frac{N'(N'-1)}{2}\cdot\eps
  \leq 1+\frac{N'(N-1)}{2}\cdot\eps
  \quad\text{ and thus }\quad
  \sum _{i=1}^{N'} p_i = 1 + \frac{N'(N-1)}{2}\cdot\tilde\eps
$$
for some $\tilde\eps \leq \eps$. We now set $p(\hat{s})$ to be $p(\hat{s} = s_i): = p_i - (N-1)\tilde\eps/2\geq p_i-(N-1)\eps/2 \geq 0$ for all $i\leq N'$.
Now consider an opening strategy $\dopen _{PQ}$ and let $p(s)$ be
the resulting output distribution. By definition of the $p_i$, it follows that
$p(s = s_i) \leq p_i$ for all $i \leq 2^n$, and $p_i \leq p(\hat{s} = s_i) + (N-1)\eps/2$ for all $i \leq N'$. 
By Lemma~\ref{lm:eps_dist}, we can conclude that there exists a consistent joint distribution $p(\hat{s}, s)$ with $p(\hat{s} = s = s_i) = \min\set{p(s=s_i),p(\hat{s}=s_i)} \geq p(s=s_i) - (N-1)\eps /2$ for all $i\leq N'$, and thus $p(\hat{s} \neq s = s_i) = p(s = s_i)  - p(\hat{s} = s = s_i) \leq (N-1)\eps /2$ for all $i\leq N'$
Furthermore, when
$N' < i \leq N$, we have $p(\hat{s} \neq s = s_i) = p(s=s_i) \leq p_i < (N-1)\eps/2$ by
definition of~$N'$. Since the $p_i$ are sorted in descending order, it follows
that for all $i > N$
$$
  p(\hat{s} \neq s = s_i) = p(s = s_i) \leq p_i \leq p_N \leq \frac{1}{N} \sum _{i=1}^N p_i \leq 
  \frac{1}{N} + \frac{N-1}{2}\cdot\eps
$$
and thus, we have shown for all $s_\circ\in \bin ^n$ that
$$
  p(\hat{s} \neq s = s_\circ) \leq \frac{1}{N} + \frac{N-1}{2}\cdot \eps\text{.}
$$
We now select $N$ so that this value is minimized: It is easy to verify that the
function $f: \mathbb R_{>0} \to \mathbb R_{>0}$, $x\mapsto 1/x + (x-1)\eps /2$
has its global minimum in $\sqrt{2/\eps}$; thus, we pick
$N := \lceil \sqrt{2/\eps}\rceil$, which gives us
$$
  p(\hat{s} \neq s = s_\circ) \leq \frac{1}{N} + \frac{N-1}{2}\cdot \eps \leq \frac{1}{\sqrt{2/\eps}} + \frac{\sqrt{2/\eps}}{2}\cdot \eps = \sqrt{2\eps}
$$
for any $s_\circ \in \bin^n$, as claimed.\qed

\section{The Hiding Property of Composed Schemes}
\label{ap:hiding}

We already mentioned that the standard hiding property is not good enough for
multi-round bit commitment schemes: The standard definition is not violated if
the verifier learns the string $s$ immediately after the commit phase. However,
the purpose of multi-round schemes is to maintain the commitment over a longer
period of time in the relativistic setting, without \em disclosing \em the
string $s$ until the very end. In this appendix, we define a hiding property
that captures this requirement, and we prove that a composed scheme
$\cal S'' = \cal S\star\cal S'$ is hiding if both $\cal S$ and $\cal S'$ are
hiding (with the error parameters adding up).

\begin{definition}
Let ${\cal S} = (\com_{PQV},\open_{PQV})$ be a commitment scheme. We write $v$
for the verifier's view immediately before the last round of communication in
$\open_{PQV}$. We say that a scheme is \em $\eps$-hiding until the last
round \em if for any (possibly dishonest) verifier $V$ and any two inputs $s_0$
and $s_1$ to the honest provers, we have $d(p(v|s_0), p(v|s_1))\leq\eps$.
\end{definition}

\begin{theorem}
\label{thm:add_hiding}
Let $\cal S$ be a $\eps$-hiding commitment scheme and $\cal S'$ a scheme that is
$\delta$-hiding until the last round. If $(\cal S, \cal S')$ is eligible, then
the composed scheme $\cal S'' = \cal S\star \cal S'$ is
$(\eps + \delta)$-hiding until the last round.
\end{theorem}

\begin{proof}
Fix a strategy against the hiding-until-the-last-round property of
$\cal S''$. We consider the distribution $p(v,y,v'|s)$ where $s$ is the string
that the provers commit to, $v$ the verifier's view after $\dcom_{PQV}$ has been
executed, $y$ the opening information to which $Q$ commits using the scheme
$\cal S'$, and $v'$ the verifier's view immediately before the last round of
communication. We need to show that
$d(p(v'|s_0), p(v'|s_1)) \leq \eps + \delta$ for any $s_0$ and $s_1$.

First, note that $p(v'|v,y,s_b) = p(v'|v,y)$ since $v'$ is produced by $P$, $Q$
and $V$ acting on $y$ and $v$ only. From any strategy against
$\cal S''$, we can obtain a strategy against $\cal S'$ by fixing $v$. Thus,
by the hiding property of $\cal S'$, for any $y_0$ and $y_1$, we have
$
  d(p(v'|v,y=y_0), p(v'|v,y=y_1))\leq\delta
$
and it follows by the convexity of the statistical distance in both arguments
that
$$
  p(v'|v,s_0) = \sum_y p(y|v,s_0)p(v'|v,y)
  \approx_{\delta} \sum_y p(y|v,s_1)p(v'|v,y)
  = p(v'|v,s_1)
$$
where we use $\approx_\delta$ to indicate that the two distributions have
statistical distance at most $\delta$. Since we have
$d(p(v|s_0), p(v|s_1))\leq\eps$ by the hiding property of $\cal S$, it follows
that
$$
  p(v'|s_0) = p(v,v'|s_0) = p(v|s_0)p(v'|v,s_0)
  \approx_{\delta} p(v|s_0)p(v'|v,s_1)
  \approx_{\eps} p(v|s_1)p(v'|v,s_1) = p(v,v'|s_1) = p(v'|s_1)
$$
where the first and last equality hold because $v'$ contains $v$ since $v'$ is
the view of $V$ at a later point in time.
\qed
\end{proof}

\end{appendix}

\end{document}